\documentclass[a4paper,11pt]{article}

\usepackage{geometry}

\usepackage{amssymb}
\usepackage{dsfont} % support for indicator function
\usepackage{amsthm}
\usepackage{harvard}
\usepackage{amsopn}
\usepackage{amsmath}
\usepackage{booktabs,caption,fixltx2e}
\usepackage{hyperref}
\usepackage{adjustbox}
\usepackage{multirow}
\usepackage{subfigure}
\usepackage{color, colortbl}
\usepackage{xurl}
\usepackage{placeins}
\usepackage{setspace}

\usepackage{comment}

\newcommand{\ds}{\displaystyle}

\newcommand{\R}{\mathbb{R}}

\newcommand{\E}{\mathbb{E}}
\newcommand{\N}{\mathbb{N}}

\newcommand{\D}{\mathcal{D}}

\newcommand{\one}{\mathds{1}}

\newcommand{\LGD}{{\rm LGD}}
\newcommand{\PD}{{\rm PD}}
\newcommand{\ELGD}{{\rm ELGD}}
\newcommand{\VLGD}{{\rm VLGD}}

\newtheorem{theorem}{Theorem}[section]

\newtheorem{lemma}[theorem]{Lemma}

\usepackage{algorithm2e}

\title{Measuring Name Concentrations through Deep Learning}

\author{
	Eva L\"{u}tkebohmert$^1$\footnote{Corresponding author, email: eva.luetkebohmert@finance.uni-freiburg.de}
	\, and Julian Sester$^2$
}

\begin{document}
	\onehalfspacing
	\maketitle
	\vspace{-6ex}
	
	\begin{center}
		\small\textit{$^{1}$Department of Quantitative Finance,\\
			Institute for Economic Research, University of Freiburg,\\
			Rempartstr. 16,	79098 Freiburg, Germany.\\[2mm]
		$^{2}$National University of Singapore, Department of Mathematics,\\ 21 Lower Kent Ridge Road, 119077}            
	\end{center}

	\begin{abstract}
    %As institutions that predominantly lend to sovereign borrowers to achieve their development goals, Multilateral Development Banks' (MDBs) loan portfolios typically consist of a small number of borrowers. The approaches currently applied by rating agencies to account for this exposure to single name concentration risk in the assessment of MDBs' capital adequacy have been criticized for being overly conservative and thereby restricting MDBs' lending headroom. 
    We propose a new deep learning approach for the quantification of name concentration risk in loan portfolios. Our approach is tailored for small portfolios and allows for both an actuarial as well as a mark-to-market definition of loss. The training of our neural network relies on Monte Carlo simulations with importance sampling which we explicitly formulate for the CreditRisk$^{+}$ and the ratings-based CreditMetrics model. Numerical results based on simulated as well as real data demonstrate the accuracy of our new approach and its superior performance compared to existing analytical methods for assessing name concentration risk in small and concentrated portfolios. 
	\end{abstract}

%% PACS codes here, in the form: 
%\PACS code \sep code

%% MSC codes here, in the form: 
%\MSC code \sep code
%% or \MSC[2008] code \sep code (2000 is the default)

%% \linenumbers

%% main text

\section{Introduction}

Specialized institutions often hold loan portfolios with a limited number of borrowers, making them highly exposed to name concentration risk, i.e. to the undiversified idiosyncratic risk associated with the default of individual borrowers.
Prominent examples are Multilateral Development Banks (MDBs), which lend to developing countries to enhance economic growth, reduce poverty and adapt to and mitigate the effects of climate change. Due to their development mandate, MDBs' loan portfolios typically consist of only a small pool of sovereign borrowers, leaving them particularly vulnerable to single name concentration risk.
For instance, the World Bank's International Bank for Reconstruction and Development (IBRD) had only 78 sovereign borrowers as of June 2022, while the regional MDBs have even fewer. %Check financial statements for these numbers or refer to \cite{Humphrey2018} for similar numbers for 2017.
%Thus, MDB portfolios are subject to substantial single name concentration risk which is inevitable due to their development mandate. 
%In view of this, name concentrations play a crucial role in the risk assessment of MDBs.
%Name concentration risk can be quantified as the difference between the Value-at-Risk (VaR) of the true portfolio loss distribution and the asymptotic loss distribution when all idiosyncratic risk is diversified away. 
While existing analytical methods for quantifying name concentration risk are highly accurate for typical commercial bank portfolios, their precision declines significantly when applied to portfolios with fewer than one hundred borrowers, such as those usually held by MDBs (compare \cite{LuetkebohmertSesterShen2023}).

In this paper, we propose a new deep learning (DL) approach for quantifying undiversified idiosyncratic risk. Our neural network (NN) approach is tailored for the small and highly concentrated portfolios and allows to measure concentration risk both on both an actuarial basis and a mark-to-market (MtM) basis. More specifically, we train a feed-forward NN to learn the difference between the Value-at-Risk (VaR) of the true portfolio loss distribution and the VaR of the asymptotic loss distribution when all idiosyncratic risk is diversified away. While the latter can be calculated analytically within the single-factor conditional independence framework underlying various industry models of portfolio credit risk, the VaR of the true loss distribution relies on Monte Carlo (MC) simulations. To speed up the computations, we apply Importance Sampling (IS). Thus, as a further contribution, we explicitly formulate algorithms for calculating portfolio VaR based on MC simulation with IS in two popular industry models, the actuarial CreditRisk$^{+}$ model and the rating-based CreditMetrics model. The portfolios used in the training process are simulated based on appropriate parametric distributions for the various input variables. 

We then evaluate the accuracy of our proposed method in comparison to existing analytical approximations as the granularity adjustment (GA) of \cite{GordyLuetkebohmert2013}, which is currently used in \cite{S&P2018} for the evaluation of MDB capital adequacy, and the GA for the MtM CreditMetrics model as developed in \cite{GordyMarrone2012}. We benchmark all methods with the exact calculations of name concentration risk based on MC simulations. 
First, we assess the accuracy of our NN based on (out-of-sample) simulated test portfolios where portfolio characteristics are sampled according to the same distribution as in the training phase. Our numerical results document that our methodology is much more accurate than the benchmark analytical approximations. 
Next, we perform various sensitivity analyses to investigate the impact of different portfolio characteristics on the performance of our DL approach. The fast and accurate estimates of our NN allow, in particular, to study how deviations in the creditworthiness or changes in the exposure share of individual borrowers may impact the level of name concentrations in the total portfolio. This in turn has important implications for the investment decisions of institutions and the risk assessment of loan portfolios. Comparable sensitivity evaluations using MC simulations would be much more time consuming.
We then show that the trained network successfully generalizes to unseen data which may not exactly follow the parametric distributions used to create the training dataset. To this end, we apply our NN to real portfolios that have been constructed based on the publicly available financial statements of eleven MDBs. Our results display the superior performance of our method compared to the benchmark analytical GAs also for the real portfolio data. Finally, we examine the application of our methodology to stress testing. Given that the DL approach, once trained, enables instant calculation of GAs, it is highly suited to this purpose. Our results further demonstrate the robustness of the NN-based GA across varying market conditions.   \\

Our paper relates to the broader literature on concentration risk. We compare our NN-based GA to the analytical approximation GAs of \cite{GordyLuetkebohmert2013} and \cite{GordyMarrone2012}. Both approaches build on the GA methodology developed in \cite{Gordy2003}, \cite{Wilde2001} and \cite{MartinWilde2002}, which was introduced as an adjustment to the Asymptotic Single Risk Factor (ASRF) model in \cite{Basel2001CP2}. Mathematically, the methodology builds on the work of \cite{Gourieroux}. Further contributions to this field are due to \cite{PykhtinDev2002}, \cite{Wilde2001IRB}, \cite{EmmerTasche2004} and \cite{Voropaev}. \cite{UbertiFigini2010} suggest an index to measure name concentration risk, which is simple to compute but less accurate than the GA methodology. \cite{GrippaGornicka2016} apply MC simulation to the part of the portfolio consisting of large exposures, whereas the analytic approximation is applied to the part of the portfolio consisting of small exposures. This reduces the computational burden while maintaining accuracy and flexibility. However, for the small portfolios we target in this paper, the approach would basically amount to applying the full MC simulation to the entire portfolio. The empirical work of \cite{GordyLuetkebohmert2013}, \cite{TarashevZhu2008}, \cite{Heitfieldetal2006} and \cite{GuertlerHeitheckerHibbeln} documents the accuracy of the approximate GAs for commercial bank portfolios consisting of at least several hundred obligors. Recent work of \cite{LuetkebohmertSesterShen2023}, however, shows that these methods are substantially less reliable for portfolios with fewer than one hundred borrowers. 
%\textcolor{blue}{The suitability of the approximate GA approaches for measuring name concentration risk in small and concentrated portfolios, such as MDB portfolios, has also been questioned in \cite{Perraudin2016} and \cite{Humphrey2015, Humphrey2018}.} 
Our NN-based GA is tailored exactly for this application, i.e. for very small and concentrated portfolios, where currently existing methods either fail because they rely on approximations that provide accurate results only for larger portfolios, or are computationally expensive because they require many simulations to obtain accurate estimates.
\begin{comment}
For the training of our NN-based GA we build on MC method with IS. To this end, we build on previous work of \cite{glasserman2005importance} who consider normal copula models, where borrowers are conditional independent given some normally distributed latent factor. Similar methods for the CreditRisk$^{+}$ model have been suggested in \cite{glasserman2003importance}. Other related work is due to \cite{glynn1989importance} and \cite{glasserman1999asymptotically}. All of these approaches focus on IS procedures for rare event simulations, i.e. estimates of the probability that the loss distribution is larger than some given threshold. IS methods for the estimation of quantiles have been suggested in \cite{glynn1996importance}. We follow that approach and provide explicit IS algorithms for the estimation of portfolio VaR in the CreditRisk$^{+}$ model and the MtM ratings-based CreditMetrics model.\\
\end{comment}

The results of our quantitative analyses have important practical implications for the risk evaluation of MDBs. The three major credit rating agencies (CRAs) each take different approaches to incorporating concentration risk into their evaluation criteria. While Moody's and Fitch follow a more qualitative approach, S\&P's methodology is based on the GA derived in \cite{GordyLuetkebohmert2013}. The latter has originally been developed to evaluate concentration risks in commercial banks' loan portfolios which even for relatively small portfolios are substantially larger than typical MDB portfolios and hence are much less concentrated. 
The S\&P's approach has been criticized for various reasons (see for instance \cite{Perraudin2016} and \cite{Humphrey2015,Humphrey2018}).  
It has been specifically recommended in \cite{CAFReview2022} that MDB capital adequacy
frameworks should appropriately account for concentration risk in MDB
portfolios and that rating agencies’ methodologies should be further evaluated in
this respect (compare Recommendations 1B on p. 28 and 4 on p. 41). The recent study \cite{LuetkebohmertSesterShen2023} addresses this task and shows that the S\&P's approach can substantially overestimate name concentration risk in MDB portfolios. This, in turn, may significantly reduce MDBs' overall lending headroom. Indeed, theoretical work by \cite{BakerWurgler2015} indicates that stricter capital requirements may lead to an increase in banks' costs of capital which in turn decreases their lending capacity. Empirical work in \cite{PeekRosengren1997,HoustonJamesMarcus1997,BernankeLown1991,HancockLaingWilcox1995, BrinkmannHorvitz1995}, suggests that banks adjust their lending after capital shocks. \cite{KashyapSteinHanson2010} provide evidence that higher capital requirements may lead to a reduction in lending and a shift of lending activities to the shadow banking sector. Although MDBs are unregulated institutions and therefore not subject to specific capital requirements, they rely heavily on bond markets for wholesale funding (as opposed to commercial banks which have deposits as funding source). This creates a strong incentive for them to maintain high credit ratings, which in turn imposes conditions on capital they must hold, particularly to mitigate name concentration risk -- a major and barely diversifiable risk factor, especially given their development mandate.
\cite{Humphrey2017} shows how the assessment of MDBs' financial strength by rating agencies has a demonstrably significant impact on their scope for lending and thus on the achievement of their development goals. 
%This applies in particular to the evaluation of name concentration risk in MDB portfolios, which represents a major risk factor in view of their development-related sovereign loan portfolios (which are generally even concentrated on certain regions).
The important role of rating agencies for access to the capital markets was also pointed out in a more general context, for example in \cite{Kruck2016}. Further, \cite{bolton2012credit} and \cite{becker2011did} show how competition among rating agencies can reduce market efficiency. With regard to MDBs, \cite{Settimo2019} suggests that the rating methodologies should take more appropriate account of the special features of MDB portfolios, which would then lead to greater scope for MDB lending.
Our novel DL-based approach offers a solution to this problem in providing a fast and accurate methodology that is tailored to small and concentrated portfolios such as MDB portfolios. Our empirical results document that the approach is very accurate and suggests a substantial reduction in capital adequacy requirements. In this way and in light of the above mentioned literature, our accurate NN-based GA may lead to an increased lending capacity for MDBs, which ultimately can be used to support physical and social development goals.

The paper is structured as follows. In Section \ref{sec measures of concentration risk} we define an exact measure for single name concentration risk in loan portfolios under both an actuarial and a MtM setting. Section \ref{sec IS} addresses the simulation of portfolio VaR in both the CreditRisk$^{+}$ and the CreditMetrics models using IS. We then develop a DL-based approach to assess name concentration risk and prove its approximation property in Section \ref{sec deep learning}. Section~\ref{sec performance evaluation} examines the sensitivity of the GA to various input parameters, evaluates our proposed NN-based GA in comparison to exact MC-based and approximate GA calculations using realistic MDB portfolios, and investigates its application to stress testing. Finally, Section \ref{sec conclusion} concludes. 
%The Appendix contains a review of existing analytical GA methods and some auxiliary mathematical result.

\section{Quantifying Portfolio Concentration Risk}\label{sec measures of concentration risk}

Denote current time by $t=0$ and suppose we want to measure portfolio concentration risk over a time period of length $T>0$. 
Consider a portfolio of $N$ obligors and suppose all loans have been aggregated on obligor level so that there is only one loan to each obligor.
Denote by $L_n$ the loss rate of obligor $n$ and by $L  \in [0,1]$ the total portfolio loss ratio $L=\sum_{n=1}^N L_n $. 
Let $\alpha_q(Y)=\inf\{y\in\mathbb{R}: \mathbb{P}(Y\geq y)\leq 1-q\}$ denote the Value-at-Risk (VaR) at level $q\in (0,1)$ of a random variable $Y$ defined as the $q^{\rm th}$ quantile of the distribution of $Y$.
%If the distribution function $F_Y$ of $Y$ is strictly increasing and continuous, then $\alpha_q(Y)=F_Y^{-1}(q)$, the ordinary inverse of the distribution function (compare \cite{McNeilFreyEmbrechts2015}).
We consider a single factor model and denote by $X$ its realization at the time horizon~$T$. Further, we assume a standard conditional independence framework, i.e.\,conditional on the systematic risk factor $X$ all loss variables $L_n$ are independent across borrowers $n=1,\ldots,N$. When the number $N$ of borrowers in the portfolio grows $N\rightarrow \infty$, the portfolio becomes infinitely fine grained and $|L-\mathbb{E}[L|X]|\rightarrow 0$ almost surely. As shown in \cite{Gordy2003}, this implies that $|\alpha_q(L) - \alpha_q(\mathbb{E}[L|X])|\rightarrow 0$ as $N\rightarrow \infty$ and under some mild regularity conditions it holds that $\alpha_q(\mathbb{E}[L|X])=\mathbb{E}[L|\alpha_q(X)]$. The undiversified idiosyncratic risk of the portfolio is then described by the difference 
\begin{equation}\label{equ def CR}
GA_q^{\rm exact}(L)\equiv \alpha_q(L)- \mathbb{E}[L|\alpha_q(X)],
\end{equation}
which we label the \textit{exact GA} at level $q$ for single name concentration risk in the loan portfolio.
To make this expression more explicit, we have to specify how loss is measured. 
Therefore, first note that the above definition of concentration risk can be applied to any risk factor model of portfolio credit risk, no matter whether losses are measured on actuarial or MtM basis. 
%The IRB approach that underpins the computation of regulator capital requirements under the Basel framework can be viewed as a mark-to-market (MtM) approach due to its maturity adjustment (compare \cite{GordyLuetkebohmert2010}, Sec. 1). However, the maturity adjustment is of course only a rough approximation. 
In the following, we will discuss these two alternative definitions of loss in more detail and formulate the corresponding exact GAs in both approaches.

\subsection{Actuarial Approach}\label{sec actuarual}

Let $A_n$ denote the (deterministic) exposure at default and $a_n=A_n/\sum_{i=1}^N A_i$ the exposure share of borrower $n$. Denote by $U_n$ the loss rate of borrower $n$. Under an actuarial definition of loss, we can define $U_n=\LGD_n\cdot D_n,$ where $D_n$ denotes the default indicator equal to 1 if the borrower defaults and 0 otherwise. The portfolio loss rate $L$ can then be expressed as
\begin{equation}\label{eq_defn_L_actuarial}
L=\sum_{n=1}^N a_n U_n=\sum_{n=1}^N a_n \LGD_n D_n.
\end{equation} 
In line with the CreditRisk$^{+}$ model, we assume that conditional on a realization $X=x$ of the systematic risk factor $X$, defaults are independent across borrowers and conditional default probabilities are given as
\begin{equation}\label{equ cond PD actuarial}
\pi_n(x)=\PD_n (1+\omega_n(x-1)),
\end{equation}
where $\PD_n$ denotes the unconditional default probability of borrower $n$ and $\omega_n$ is the sensitivity of borrower $n$ to the systematic risk factor $X$. We assume $X$ to be Gamma-distributed with mean 1 and variance $1/\xi$ for some $\xi>0$. Given a realization $X=x$ of the risk factor, the default indicators $D_n$ of obligors $n=1,\ldots, N$ are conditionally independent Bernoulli random variables with parameter $\mathbb{P}(D_n=1|X=x)=\pi_n(x).$ 
%In order to obtain an analytic solution, the distribution of the default indicator is approximated by a Poisson distribution in the CreditRisk$^+$ model.
We assume the loss given default $\LGD_n$ is a beta-distributed random variable with mean $\ELGD_n$ and variance $\VLGD_n^2=\nu\cdot \ELGD_n (1-\ELGD_n)$ for some fixed volatility parameter $\nu>0$ and independent of the systematic risk factor $X$. 

Thus, the exact GA given in (\ref{equ def CR}) can be expressed as
\begin{equation}\label{equ GA exact actuarial}
    GA^{\rm exact,act}_q(L)= \alpha_q\left(\sum_{n=1}^N a_n \LGD_n  D_n\right)- \sum_{n=1}^N a_n \ELGD_n \pi_n(\alpha_q(X)).
\end{equation}
While the second term is in closed form, the first term needs to be calculated by MC simulation. We present an IS approach for its efficient computation in Section \ref{sec IS CR+}. The GA of \cite{GordyLuetkebohmert2013} is an analytical approximation for expression (\ref{equ GA exact actuarial}), which we briefly review in the e-companion \ref{app actuarial GA}.

\subsection{Mark-to-Market Approach}\label{subsec MtM}

Denote by $A_n$ the current market value of the exposure to obligor $n$ 
%In an actuarial setting this would be the face value of the instrument
and by $a_n=A_n/\sum_{i=1}^N A_i$ the portfolio weight of this exposure. We assume that all positions in the portfolio consist of defaultable coupon bonds. Then, $A_n$ refers to the notional amount outstanding.  
Similarly to the actuarial setting, the borrower specific loss given default variable $\LGD_n$ is modelled as a beta-distributed random variable with mean $\ELGD_n$ and variance $\VLGD_n^2=\nu\cdot \ELGD_n (1-\ELGD_n)$ for some fixed volatility parameter $\nu>0$ and independent of all other sources of risk.

%Accordingly, the central third moment of the beta-distributed random variable $\LGD$ can be computed as 
%\begin{equation}\label{equ third moment LGD}
%\mathbb{E}\left[(\ELGD-\LGD)^3\right] = - \frac{2(1-2\ELGD)}{1+\tfrac{1}{\nu}} \cdot \VLGD^2.
%\end{equation}

Following \cite{GordyMarrone2012} we define the loss on borrower $n$ in the MtM setting as the difference between expected return $\mathbb{E}[R_n]$ and realized return $R_n$ on position~$n$, discounted to the current date at the (continuously compounded) risk-free interest rate~$r$, where return $R_n$ is defined as the ratio of the market value at time $T$ to the current time $0$ market value. Thus, the loss on borrower $n$ is $L_n=(\mathbb{E}[R_n]-R_n)e^{-rT},$ which is expressed as a percentage of current market value of exposure to borrower $n$. Hence, the total portfolio return is $R=\sum_{n=1}^N a_n R_n$ and the portfolio loss rate equals
\begin{equation}\label{equ portfolio loss MtM}
L=\left(\mathbb{E}[R]-R\right)\cdot e^{-rT}.
\end{equation}
We apply a ratings-based approach so that the state of obligor $n$ at time $T$ can only take a finite number $S$ of possible states $S_n\in\mathcal{S}=\{0,1,\ldots, S\}.$ State $0$ is the absorbing default state and state $S$ dentoes the highest possible rating (e.g. AAA rating in S\&P's rating classification).
Denote by $p_{gs}$ the transition probability from current rating $g$ to rating $s$ at time $T$. 
The latent asset return $Y_n$ associated with borrower $n$ is modelled as
$$
Y_n=\sqrt{\rho_n} X+\sqrt{1-\rho_n} \epsilon_n,
$$
where the single systematic risk factor $X\sim\mathcal{N}(0,1)$ %, so that $h(x)=\phi(x)$ is the standard normal density, 
and $\rho_n$ determines the asset correlation of borrower $n$. The idiosyncratic risk factors $\epsilon_n$ are assumed to be independent and identically distributed $\mathcal{N}(0,1)$ as well as independent of $X$. Hence, asset returns are independent across borrowers given the systematic risk factor and the conditional probability that obligor $n$ is in state $S_n=s$ at time $T$ given $X=x$ equals
\begin{equation}\label{eq_pi_def}
\pi_{ns}(x)=\Phi\left(\frac{C_{g(n),s}-x\sqrt{\rho_n}}{\sqrt{1-\rho_n}}\right)-\Phi \left(\frac{C_{g(n),s-1}-x\sqrt{\rho_n}}{\sqrt{1-\rho_n}}\right).
\end{equation}
Here, $C_{g(n),s}$ and $C_{g(n),s-1}$ denote the threshold values so that borrower $n$ with current rating $g(n)$ is in rating $s$ at time $T$ when $C_{g(n),s-1}<Y_n\leq C_{g(n),s}$. The borrower defaults when $Y_n\leq C_{g(n),0}$ and we set $C_{g(n),-1}=-\infty$ and $C_{g(n),S}=+\infty.$ The thresholds are determined from the transition probabilities $p_{sg}$ as
\begin{equation}\label{equ thresholds C_gs}
C_{g,s}=\Phi^{-1}\left(\sum_{i=0}^s p_{g,i}\right),\quad\mbox{for}\quad s=0,\ldots,S-1.
\end{equation}
Note that all borrowers with the same current rating $g$ face the same thresholds for transitioning to rating $s$ at time $T$. 

We adopt the CreditMetrics setting as in \cite{GordyMarrone2012} and suppose position $n$ consists of a bond with face value 1, maturity in $\tau>T$ years, coupon payments of $c_n \delta$ at times $t_1,\ldots,t_{m}=\tau$ with accrual period $\delta=t_i-t_{i-1}$ (expressed as fraction of year), and current market value $P_{n0}$.\footnote{The position is afterwards scaled by the exposure share $a_n$ when determining the portfolio loss rate. The total market value of position $n$ is then given by $A_n P_{n0}$.}
The forward value at time $t>0$ of the analogous default-free bond is
$$
F(t)=c_n\delta\sum_{k=1}^m \one_{\{t_k\geq t\}} e^{-r(t_k-t)}+e^{-r(t_m-t)}.
$$
When borrower $n$ is in rating grade $g(n)$ at time $0$, the value of the defaultable bond $n$ at current time $0$ equals
\begin{align}\label{equ P_n0}
P_{n0}&=F(0)-\sum_{k=1}^m e^{-rt_k} (p_{g(n)}^*(0,t_k)-p_{g(n)}^*(0,t_{k-1}))
\cdot \ELGD_n\cdot  F(t_k),
\end{align}
assuming that, in the event of default, the bondholder claims the no-default value of the bond (compare \cite{HullWhite1995}, \cite{JarrowTurnbull1995}). %\footnote{Alternatively, one could consider the case where the bondholder claims the principal payment in case of default.} 
Here, $p_g^*(u,v)$ denotes the risk-neutral probability that an obligor in rating grade $g$ at time~$u$ defaults before time $v$. 
Moreover, the value of the bond $n$ at time $T$ conditional on a non-default state $S_n=s$ equals
%\begin{equation}
\begin{align}
\label{equ P_nT}
P_{nT}(s)&= c_n\delta \sum_{k=1}^m \one_{\{t_k< T\}} e^{r(T-t_k)} + F(T)\\[4mm]
& -\sum_{k=1}^{m} \one_{\{t_k>T\}} e^{-r(t_k-T)} \left(p^*_s(T,t_k)-p^*_s(T,t_{k-1})\right) \cdot \ELGD_n \cdot F(t_k)\nonumber
\end{align}
%\end{equation}
%where $ \mathbb{E}_T[\LGD_n]=\ELGD_n$ is the time $T$ expectation of the LGD on position $n$. {\textcolor{red}{We assume here that, given information at time $0$, the distribution of $ \mathbb{E}_T[\LGD_n]$ is identical to the distribution of $\LGD_n$. However, since $P_{nT}(s)$ is the value of the bond at time $T$ conditional on a non-default state $s$, only the expected LGD enters the above formula.}} 
The first term in equation (\ref{equ P_nT}) is the compounded value of all coupons up to time $T$, the second term is the forward value of the analogous default-free bond at time $T$ and the third term represents the discounted expected value of all cash flows later than $T$ if the borrower defaults up to maturity $t_m$ of the bond. For simplicity, we assume here that the time horizon $T$ coincides with a coupon payment date $t_k$. Otherwise, the expression $p^*_s(T,t_{k-1})$ should be replaced by 0 if $T>t_{k-1}$.
The value of the bond conditional on the borrower being in the default state $s=0$ at time $T$ can be expressed as
\begin{align}\label{equ P_nT(0)}
P_{nT}(0)=c_n\delta\sum_{k=1}^m \one_{\{t_k<T\}} e^{r(T-t_k)}+(1-\LGD_n)F(T),
\end{align}
assuming that the obligor defaults just before the time horizon $T$. Note that we consider a realization of the random variable $\LGD_n$ here instead of $\ELGD_n$ since $P_{nT}(0)$ is conditional on the default state.
The return on position $n$ is then given by
\begin{equation}\label{equ return R_n}
\begin{aligned}
    R_n&= \sum_{s=0}^S \frac{P_{nT}(s)}{P_{n0}}\cdot\one_{\left\{C_{g(n),s-1}< ~Y_n~\leq C_{g(n),s}\right\}}.
\end{aligned}
\end{equation}

The conditional expectation of the individual and total portfolio return given a realization $X=x$ of the systematic risk factor can be expressed as
\begin{align*}
    \mu_n(x)\equiv \mathbb{E}[R_n|X=x]\quad \mbox{and}\quad
    \mu(x)=\mathbb{E}[R|X=x]=\sum_{n=1}^N a_n \mu_n(x)
\end{align*}

Note that when calculating the GA through equation (\ref{equ def CR}), the expected return $\mathbb{E}[R]$ in the expression of the portfolio loss rate (\ref{equ portfolio loss MtM}) cancels out. Moreover, in the presented setting $-X$ is the systematic risk factor (as highly negative realizations of $X$ correspond to bad scenarios). Thus, we obtain the following expression for the exact GA in the MtM setting
%\begin{equation}\label{equ GA exact MtM}
%{\begin{footnotesize}
\begin{align}
    &GA^{\rm exact,MtM}_q(L)= \alpha_q(L)- \E[L~|~\alpha_{q}(-X)]  = \alpha_q(L)- \E[L~|~\alpha_{1-q}(X)]\label{equ GA exact MtM} \\
    =& e^{-rT}\bigg[ \sum_{n=1}^N a_n \cdot \mu_n(\alpha_{1-q}(X)) +\alpha_q\left(-\sum_{n=1}^N a_n \sum_{s=0}^S \frac{P_{nT}(s)}{P_{n0}}\cdot\one_{\left\{C_{g(n),s-1}< ~Y_n~\leq C_{g(n),s}\right\}}\right) \bigg].\nonumber
\end{align}
%\end{footnotesize}}
%\end{equation}
The first term can be further simplified under some mild assumptions. Therefore, following \cite{GordyMarrone2012}, we decompose the conditional expected return as
\begin{equation}\label{equ mu_n(x)}
\mu_n(x)=\mathbb{E}[R_n|X=x]=\sum_{s=0}^S \lambda_{ns}(x) \cdot \pi_{ns}(x)=\langle \Pi_n(x),\Lambda_n(x)\rangle,
\end{equation}
where $\pi_{ns}(x)\equiv \mathbb{P}(S_n=s|X=x)$ and $\lambda_{ns}(x)\equiv \mathbb{E}[R_n|X=x,S_n=s],$ while
$\Pi_n(x)$, resp. $\Lambda_n(x)$, is the vector of $\{\pi_{ns}(x)\}$, resp. of $\{\lambda_{ns}(x)\}.$
Moreover, in line with the CreditMetrics model, we assume that market credit spreads at the time horizon $T$ are solely functions of the rating grade, i.e.
\begin{equation}\label{equ lambda_ns}
\lambda_{ns}(x)\equiv \mathbb{E}[R_n|X=x,S_n=s]=\mathbb{E}[R_n|S_n=s]\equiv \lambda_{ns}.
\end{equation}
Hence, we can express the expected return on position $n$ given $X=x$ as
\begin{equation}\label{equ expected return given x}
    \mu_n(x)=\sum_{s=0}^S \lambda_{ns}\cdot \pi_{ns}(x).
\end{equation}
At time $0$, the expected return on bond $n$ given that borrower $n$ is in rating grade $S_n=s>0$ at horizon $T$ is given by
$$
\lambda_{ns}=\mathbb{E}[R_n|S_n=s]=\mathbb{E}[P_{nT}(s)]/P_{n0}.
$$
Note that since we assume $\mathbb{E}_T[\LGD_n]=\ELGD_n$, the expected value of bond $n$ conditional on state $S_n=s$ equals $\mathbb{E}[P_{nT}(s)]=P_{nT}(s).$
To calculate the latter, we use the relationship $p^*_s(T,t)=p^*_s(0,t-T)$ for $t\geq T,$ which follows from the Markovian transition model (see \cite{HullWhite2000}). Thus, we obtain
%{\begin{footnotesize}
\begin{equation}\label{eq_defn_lambda_ns}
\begin{aligned} 
\lambda_{ns}&=\frac{1}{P_{n0}} \Big( c_n\delta \sum_{k=1}^m \one_{\{t_k< T\}} e^{r(T-t_k)}+F(T)\\[4mm]
& \hspace{5mm} -\sum_{k=1}^m \one_{\{t_k>T\}}e^{-r(t_k-T)} \left(p^*_s(0,t_k-T)-p^*_s(0,t_{k-1}-T)\right)\cdot  \ELGD_n \cdot F(t_k)\Big).
\end{aligned}
\end{equation} 
%\end{footnotesize}}
Similarly, the expected return in the default state is given by
$$
\lambda_{n0}=\frac{1}{P_{n0}} \left(c_n\delta\sum_{k=1}^m \one_{\{t_k<T\}} e^{r(T-t_k)}+(1-\ELGD_n)F(T)\right),
$$
assuming that default occurs just prior to the time $T$.
%, and {\color{red} We don't need the variance of returns for the MC approach, right? I think we included the next term solely for the analytical approximation GA in the MtM setting. I copied the below expression to the end of Appendix A.2.}
%$$
%\xi_0^2=\left(\frac{F(T)}{P_{n0}}\right)^2 \VLGD^2.
%$$
Plugging the expressions for $\lambda_{ns}$ and $\lambda_{n0}$ into equation (\ref{equ expected return given x}) and replacing the risk factor by $\alpha_{1-q}(X)$, we obtain an explicit representation of the term $\mu_n(\alpha_{1-q}(X))$ in equation (\ref{equ GA exact MtM}). The second term needs to be computed by MC simulation. We provide an efficient IS approach for this task in Section \ref{sec IS CM}. \cite{GordyMarrone2012} derive an analytical approximation of the expression (\ref{equ GA exact MtM}), which we briefly review in \ref{app MtM GA}.

%\begin{remark}\label{remark FRN}
 %   Note that the above framework can also be adjusted to the case where position $n$ consists of a floating rate note (FRN) paying a fixed spread $c_n\delta$ over LIBOR at the coupon dates $t_1,\ldots,t_m=\tau$ with accrual period $\delta$. Since the value of a FRN paying the market reference rate is always reset to its par value at the coupon dates, the value of the (defaultable) FRN at time $0$ and at the coupon dates can be expressed as par value $1$ plus the value of a defaultable coupon bond paying the spread $c_n\delta$ as coupons. The value of the latter can be expressed as $P_{n0}$ at time $0$, resp. by $P_{nT}(s)$ at time $T$ conditional on state $S_n=s$ with the notation from above. For simplicity, we assume here that time $T$ is also a coupon payment date.
%\end{remark}

\section{Importance Sampling for Portfolio VaR Calculations}\label{sec IS}

%Consider the loss rate $L=\sum_{n=1}^N L_n$ of a portfolio consisting of $N$ loans with individual loss ratios $L_n$ for $n=1,\ldots, N,$ that are conditionally independent given the realization of a systematic risk factor $X$. 
%Suppose $X$  has a continuous distribution function with density $f$. 
To obtain stable results for the probability that the portfolio loss rate $L$ exceeds a certain (high) threshold $y$ or for the VaR $\alpha_q(L)$ at a high level $q\in(0,1)$ using standard MC methods requires a very large number of simulations. IS can significantly reduce the computational burden by shifting the distribution so that the rare event that we want to simulate is more likely under the new distribution than under the original density, see also, e.g. \cite{glasserman1999asymptotically}, \cite{glasserman2005importance}, \cite{glynn1989importance} and \cite[Section 11.4]{mcneil2015quantitative}. The main task consists in finding the IS density which can be achieved by \emph{exponential tilting}. 
%To this end, we consider the moment generating function (mgf) of a random variable $X$ at %$t\in\mathbb{R}$ given by
%$$
%M_X(t)=\mathbb{E}[e^{tX}]=\int_{-\infty}^{\infty} e^{tx} f(x) dx,
%$$
%where $f$ is the density of $X$. For finite mgf, exponential tilting with $t$ then defines a new %IS density
%$$
%g_t(x):= e^{tx} f(x)/M_X(t)
%$$
%and the likelihood ratio can be expressed as $r_t(x)=f(x)/g_t(x)=M_X(t) e^{-tx}.$ \\
While the estimation of rare event probabilities via IS is well studied in the literature for the normal copula model, the application to standard industry models of portfolio credit risk is less well document, in particular for VaR calculations. In this section, we fill this gap and provide explicit algorithms for VaR calculations via IS in two widely applied industry models, namely the CreditRisk$^+$ model and the ratings-based CreditMetrics model.

\subsection{VaR Calculations in CreditRisk$^+$ Model}\label{sec IS CR+}

We consider the CreditRisk$^+$ model, described in Section~\ref{sec actuarual}, where the systematic risk factor $X$ follows a  $\operatorname{Gamma}(\xi,1/\xi)$ distribution for some $\xi >0$ with a density function denoted by
$f_X(x)=\xi^\xi x^{\xi-1} e^{-x\xi}/\Gamma(\xi).$
Denote by $M_X$ the moment generating function (mgf) of $X$. Exponential tilting with $t$ then defines a new density
$$
g_t(x)=e^{tx}f_X(x)/M_X(t)=e^{tx}\frac{\xi^\xi}{\Gamma(\xi)} x^{\xi-1} e^{-x\xi}(1-t/\xi)^\xi
$$
which is the density of a Gamma$(\xi,\frac{1}{\xi -t})$ distribution for $t<\xi$. Thus, the change of distribution is defined through the likelihood ratio
$$
r_t(x)=\frac{f_X(x)}{g_t(x)}=e^{-tx}M_X(t)=e^{-tx-\xi\log(1-t/\xi)}.
$$
This first step changes the distribution of the systematic risk factor $X$.

In the second step, consider the portfolio loss ratio $L=\sum_{n=1}^N  a_n \LGD_n Y_n$ as defined in \eqref{eq_defn_L_actuarial} with $Y_n$ being Bernoulli distributed conditional on the risk factor $X$ with parameter $\pi_n(X)=\PD_n (1+\omega_n(X-1)).$
Exponential tilting of $a_n \LGD_nY_n$ with some $\tau\in\mathbb{R}$ defines a change of distribution through the likelihood ratio (conditional on $X=x$)
\begin{align*}
\tilde{r}_{\tau}(a_n \LGD_nY_n;x)&=e^{-\tau a_n \LGD_n Y_n} M_{a_n \LGD_nY_n}(\tau)\\
&=e^{-\tau a_n \LGD_n Y_n} \left(e^{\tau a_n \LGD_n} \pi_n(x)+1-\pi_n(x)\right).
\end{align*}
Since $Y_n$ are conditionally independent given $X=x$, we obtain the likelihood ratio for $L$ conditional on $X=x$ as
$$
\tilde{r}_{\tau}(L;x)=\exp(-\tau L) M_L(\tau)=\exp(-\tau L) \prod_{n=1}^N \left(e^{\tau a_n \LGD_n}\pi_n(x)+1-\pi_n(x)\right).
$$
%Since $\log(1+x)\approx x$ for small $x$ and the default probabilities are typically small, we can approximate the mgf of the Bernoulli distribution by a Poisson distribution so that
%\begin{align}
%M_{a_n \LGD_nY_n}^{Bern}(\tau)&=\left(e^{\tau a_n \LGD_n} \pi_n(x)+1-\pi_n(x)\right)=\exp\left(\log\left( 1+\pi_n(x)\left(e^{\tau a_n \LGD_n}-1\right)\right)\right)\\
%&\approx \exp\left(\pi_n(x) \left(e^{\tau a_n \LGD_n}-1\right)\right)=M_{a_n \LGD_nY_n}^{Pois}(\tau)\nonumber
%\end{align}
%the mgf of a Poisson distributed variable with parameter $\pi_n(x)$.

The likelihood of the two step change of distribution is then the product of the individual likelihood ratios, i.e.
\begin{align}\label{equ likelihood ratio bar r}
\bar{r}_{\tau,t}(L;x)&=\tilde{r}_{\tau}(L;x)\cdot r_t(x)\\
&=e^{-\tau L} \prod_{n=1}^N \left(e^{\tau a_n \LGD_n}\pi_n(x)+1-\pi_n(x)\right) e^{-tx-\xi\log(1-t/\xi)}.\nonumber
\end{align}
Next, we observe that for small probabilities $\pi_n(x)$ we can approximate the Bernoulli distribution by a Poisson distribution so that
\begin{align}
M_{a_n \LGD_nY_n}^{Bern}(\tau)&=\left(e^{\tau  a_n \LGD_n} \pi_n(x)+1-\pi_n(x)\right)\\
&=\exp\left(\log\left( 1+\pi_n(x)\left(e^{\tau a_n \LGD_n}-1\right)\right)\right)\nonumber\\
&\approx \exp\left(\pi_n(x) \left(e^{\tau a_n \LGD_n}-1\right)\right)=M_{a_n \LGD_nY_n}^{Pois}(\tau).\nonumber
\end{align}
Hence, if we set
\begin{equation}\label{equ choosing t}
t:=\sum_{n=1}^N \PD_n \omega_n(e^{\tau a_n \LGD_n}-1),
\end{equation}
then for small probabilities $\pi_n(x)$ the dependence of $\bar{r}_{\tau,t}(L;x)$ on the systematic risk factor is  eliminated (up to the Poisson approximation error). This motivates our choice of $t$, even though we decided to not apply the Poisson approximation directly, mainly since high $\PD$ values may then lead to a large approximation error.
%If $t$ is small enough, we also have $t<\xi.$ 
%Substituting in (\ref{equ likelihood ratio bar r}) gives the 2-step likelihood ratio
%$$
%\exp\left( -\tau L+\sum_{n=1}^N \PD_n (1-\omega_n)\left(e^{\tau a_n \LGD_n}-1\right)\right) e^{-k\log (1-\theta \sum_{n=1}^N \PD_n \omega_n (e^{\tau a_n \LGD_n}-1)}
%$$
%For the choice of $\tau$ we follow the idea of \emph{constant approximation} proposed in \cite[Section 5]{glasserman2005importance}.
Note that after exponential tilting the risk-factor follows a Gamma$(\xi,\frac{1}{\xi -t})$ distribution which has a mean of $\frac{\xi}{\xi-t}$. We want to achieve that the mean of the risk factor coincides with the quantile of $X$, i.e., $\frac{\xi}{\xi-t} = \alpha_q(X)$, which yields
\begin{equation}\label{equ tau_y}
\sum_{n=1}^N \PD_n \omega_n(e^{\tau a_n \LGD_n}-1) = \xi \left(1-\frac{1}{\alpha_q(X)}\right).
\end{equation}
%the independent case and set $\tau$ such that the derivative of the cumulant generating function $\Psi_{L,N}(\tau)= \Psi_{L,N}^{(1)}(\tau)+\Psi_{L,N}^{(2)}(\tau)$ of $L$ equals the threshold $y$ when we want to determine $\mathbb{P}(L>y),$ i.e.
%\begin{equation}\label{equ tau_y}
%y=\frac{d}{d\tau} \big(\underbrace{\sum_{n=1}^N \PD_n (1-\omega_n) \left(e^{\tau a_n \LGD_n}-1\right)}_{\equiv \Psi_{L,N}^{(1)}(\tau)}\underbrace{-k\log \big(1-\theta \sum_{n=1}^N \PD_n \omega_n\big(e^{\tau a_n \LGD_n}-1\big)\big)}_{\equiv \Psi_{L,N}^{(2)}(\tau)}\big).
%\end{equation}
This can be solved numerically for $\tau_y.$ After exponential tilting with $\tau_y$, the default indicator $Y_n$ is then Bernoulli distributed with parameter
$$
\frac{e^{\tau_y a_n \LGD_n}}{M_{a_n\LGD_n Y_n}^{Bern}(\tau_y)} \pi_n(X).
$$
To determine the VaR, we follow the approach outlined in \cite{glynn1996importance}, and implement the IS procedure in Algorithm~\ref{algo IS CR+} in \ref{app algorithms}.

%
%
%\subsection{CreditMetrics Default-Only Model}\label{sec IS default only CM}
%
%Consider the borrower specific loss rate
%$$
%L_n= -a_n  \left(\frac{P_{nT}(1)}{P_{n0}} \one_{\{Y_n> \Phi^{-1}\left(\PD_n\right)\}}+\frac{P_{nT}(0)}{P_{n0}}\one_{\{Y_n\leq\Phi^{-1}\left(\PD_n\right)\}}\right)
%$$
%Loss rates are conditional independent given the systematic risk factor $X$, so that
%$$
%\mathbb{P}(Y_n\leq\Phi^{-1}\left(\PD_n\right)| X=x)=\pi_{n0}(x),\;\mbox{and}\; 
%\mathbb{P}(Y_n>\Phi^{-1}\left(\PD_n\right)| X=x)=1-\pi_{n0}(x)
%$$
%with 
%$$
%\pi_{n0}(x)=\Phi\left(\frac{\Phi^{-1}(\PD_n)-\sqrt{\rho_n} x}{\sqrt{1-\rho_n}}\right).
%$$
%Thus, for $L=\sum_{n=1}^N L_n=-\sum_{n=1}^N a_n \left( \frac{P_{nT}(0)}{P_{n0}} \one_{\{Y_n\leq\Phi^{-1}\left(\PD_n\right)\}}+\frac{P_{nT}(1)}{P_{n0}} \one_{\{Y_n> \Phi^{-1}\left(\PD_n\right)\}}\right)$ the conditional mgf is
%$$
%M_{L|X}(t)=\mathbb{E}\left[e^{tL}|X\right]=\prod_{n=1}^N \left( e^{-ta_n \frac{P_{nT}(0)}{P_{n0}}} \pi_{n0}(X)+e^{-ta_n \frac{P_{nT}(1)}{P_{n0}}} (1-\pi_{n0}(X))\right)
%$$
%
%
%For $z_n\in\{0,1\}$, consider the tilted distribution
%\begin{align}
%    Q_t(\{z\}|X)&=\mathbb{E}\left[\frac{e^{tL}}{M_{L|X}(t)};\{(Z_1,\ldots,Z_N)=(z_1,\ldots,z_N)\}|X\right]\\
%    &= \prod_{n=1}^N \frac{e^{-t a_n \frac{P_{nT}(0)}{P_{n0}} z_n }e^{-t a_n \frac{P_{nT}(1)}{P_{n0}} (1-z_n) }}{ e^{-ta_n\frac{P_{nT}(0)}{P_{n0}}}\pi_{n0}(X)+e^{-ta_n\frac{P_{nT}(1)}{P_{n0}}}(1-\pi_{n0}(X))} \cdot (\pi_{n0}(X))^{z_n} (1-\pi_{n0}(X))^{1-z_n}\\
%\end{align}

\subsection{VaR Calculations in Ratings-Based CreditMetrics Model}\label{sec IS CM}
We consider the CreditMetrics framework discussed in Section~\ref{subsec MtM}. First, we consider the borrower specific loss rate
$$
L_n= -a_n \sum_{s=1}^S \frac{P_{nT}(s)}{P_{n0}} \underbrace{\one_{\{C_{g(n),s-1}<Y_n\leq C_{g(n),s}\}}}_{\equiv Z_n(s)}
$$
for $n=1,\ldots,N$ with $Z_n(s)\in\{0,1\}$ and threshold values $C_{g(n),s}$ as in (\ref{equ thresholds C_gs}). The loss rates are conditionally independent given the systematic risk factor $X$ which is assumed to be standard normally distributed. Define the vector $Z_n=(Z_n(1),\ldots, Z_n(S))\in \{0,1\}^S$ with only one $Z_n(s)=1$ for each $n$ and all other components equal to zero. We then have
$$
\mathbb{P}(Z_n=z| X=x)=\prod_{s=1}^S (\pi_{ns}(x))^{z(s)},
$$
for $z=\left(z(1),\dots,z(S)\right)\in\{0,1\}^S$, where the conditional probability $\pi_{ns}(x)$ that obligor $n$ is in state $S_n=s$ at time $T$ given $X=x$ is defined as in (\ref{eq_pi_def}).
Thus, for $L=\sum_{n=1}^N L_n=-\sum_{n=1}^N a_n \sum_{s=1}^S \frac{P_{nT}(s)}{P_{n0}} Z_n(s)$ the conditional mgf is
$$
M_{L|X}(t)=\mathbb{E}\left[e^{tL}|X\right]=\prod_{n=1}^N \mathbb{E}\left[e^{-ta_n\sum_{s=1}^S \frac{P_{nT}(s)}{P_{n0}}Z_n(s)}|X\right]=\prod_{n=1}^N \sum_{s=1}^S e^{-ta_n \frac{P_{nT}(s)}{P_{n0}}} \pi_{ns}(X).
$$
For $z_n =\left(z_n(1),\dots,z_n(S)\right)\in\{0,1\}^S$, consider the tilted distribution (conditional on $X$)
\begin{align}
    Q_t(\{z\}|X)&=\mathbb{E}\left[\frac{e^{tL}}{M_{L|X}(t)};\{(Z_1,\ldots,Z_N)=(z_1,\ldots,z_N)\}|X\right]\nonumber\\
    & =\frac{e^{-t\sum_{n=1}^N a_n \sum_{s=1}^S \frac{P_{nT}(s)}{P_{n0}} z_n(s)}}{M_{L|X}(t)} \cdot \mathbb{P}\left((Z_1,\ldots,Z_N)=(z_1,\ldots,z_N)|X\right)\nonumber\\
    &= \prod_{n=1}^N\prod_{s=1}^S \frac{e^{-t a_n \frac{P_{nT}(s)}{P_{n0}} z_n(s)}}{\sum_{s=1}^S e^{-ta_n\frac{P_{nT}(s)}{P_{n0}}}\pi_{ns}(X)} \cdot (\pi_{ns}(X))^{z_n(s)}.
\end{align}
Following the argument outlined in \cite[Section 11.4]{mcneil2015quantitative} to derive parameters for exponential tilting we set
\begin{equation}\label{eq_def_c}
c:= \mathbb{E}[L~|~\alpha_{1-q}(X)]=-\sum_{n=1}^N a_n \mu_n(\alpha_{1-q}(X)),
\end{equation}
and  our objective is to reduce the variance of our estimator which means we want to make 
%\[
\begin{align}
\mathbb{E}_g\left[\one_{\{L > c\}}r_t^2(L;X)~\middle|~X\right]&=\mathbb{E}\left[\one_{\{L > c\}}r_t(L;X)~\middle|~X\right]\nonumber\\
&=\mathbb{E}\left[\one_{\{L > c\}}M_{L|X}(t)e^{-tL}~\middle|~X\right]\nonumber
\end{align}
%\] 
as small as possible. Here, $\mathbb{E}_g[\cdot]$ denotes the expectation w.r.t. the IS density $g_t$ whereas $\mathbb{E}[\cdot]$ is the expectation w.r.t. the original density $f$. The smaller the variance of the estimator is, the more likely the event $\{L>c\}$ is under the IS density $g_t$ than under the density $f$. Moreover, we note that 
\[
\mathbb{E}\left[\one_{\{L > c\}}M_{L|X}(t) e^{-tL}~\middle|~X\right] \leq M_{L|X}(t) e^{-tc}
\]
and hence we aim at finding the value $t$ that minimizes the upper bound $M_{L|X}e^{-tc}$. Equivalently, we aim at
minimizing $\frac{d}{dt} \ln (M_{L|X}(t))-tc$ with respect to the variable $t$. This means we solve
\begin{align}
c&=\frac{d}{dt} \ln(M_{L|X}(t))=\frac{d}{dt}\left(\sum_{n=1}^N \ln\left(\sum_{s=1}^S e^{-ta_n \frac{P_{nT}(s)}{P_{n0}}}\pi_{ns}(X)\right)\right)\nonumber\\
&= - \sum_{n=1}^N \sum_{s=1}^S \frac{a_n \frac{P_{nT}(s)}{P_{n0}} e^{-ta_n \frac{P_{nT}(s)}{P_{n0}}} \pi_{ns}(X)}{\sum_{s=1}^S e^{-ta_n\frac{P_{nT}(s)}{P_{n0}}}\pi_{ns}(X)}\label{equ t(c,X) in CM model}
\end{align}
for $t$.
Next, note that, after exponential tilting the standard normally distributed risk factor $X$ possesses an $\mathcal{N}(\mu,1)$ distribution. The zero-variance distribution of the estimator for $\mathbb{P}(L>c~|~X)$ possesses a density which is proportional to $x \mapsto\mathbb{P}(L>c~|~X=x)e^{-x^2/2}$, see \cite[Sec. 5]{glasserman2005importance}. Hence, to choose a variance-reducing $\mu$, it is reasonable to consider a normal density with the same mode as the density of the zero-variance distribution, and hence to compute
$\mu$ as a solution to the optimization problem
\begin{equation}\label{eq_opt_dist}
\max_{x} \mathbb{P}(L>c|X=x) e^{-\frac{1}{2} x^2},
\end{equation}
compare also \cite{glasserman1999asymptotically}.
To find a numerical solution of \eqref{eq_opt_dist}, in line with the \emph{constant approximation} approach proposed in \cite[Sec. 5]{glasserman2005importance}, we determine $\mu$ numerically as a solution to 
\begin{equation}\label{mu_x}
\min_{x} \left\{x^2 ~|~ \mathbb{E}[L~|~X=x]>c \right\}=\min_{x} \left\{x^2 ~|~ -\sum_{n=1}^N a_n \mu_n(x)>c \right\}.
\end{equation}
The IS approach is summarized in Algorithm~\ref{algo IS CM} in \ref{app algorithms}. 
%Its reduced variance is illustrated in Figure~\ref{fig_is_1} where the convergence of the MC method with IS is plotted against a standard MC approach for the computation of the quantile.

%\begin{figure}
%\centering
%\includegraphics[scale =0.5]{pdf/plot_2_mtm_one_plot.pdf}
%\caption{The figure shows the outcome of 1 million iterations (number of iterations on the $x$-axis) of MC simulations with and without IS (according to Algorithm~\ref{algo IS CM}) for the computation of $\alpha_q(L)$ applied to a synthetic portfolio.} \label{fig_is_1}
%\end{figure}

\section{Deep Learning Name Concentration Risk}\label{sec deep learning}

Calculating the expressions (\ref{equ GA exact MtM}), resp. (\ref{equ GA exact actuarial}), via MC simulation is computationally expensive even when using IS methods, since we are interested in very high quantiles $q=99.9\%$ or even higher. Therefore, analytic approximations as those discussed in \ref{app GA methodology} have been intensively studied in the literature. 
%\cite{GordyLuetkebohmert2013} develop an analytic formula for the GA in the actuarial setting while \cite{GordyMarrone2012} provide an asymptotic approximation for the GA in the MtM setting. We briefly review both methods in Appendix \ref{app GA methodology}. 
These approaches perform well for commercial bank portfolios which even for small portfolios contain enough obligors so that the approximation error is sufficiently small. However, when portfolios are very small with less than 100 borrowers, these method can substantially overestimate the exposure to name concentration risk as documented by our numerical results in Section \ref{sec performance evaluation} and the recent study in \cite{LuetkebohmertSesterShen2023}.

Therefore, in this section, we suggest a NN-based based approach to compute the GAs for both the actuarial setting as well as the MtM setting. We train our network specifically for small portfolios using exact GA computations derived via MC simulations with IS. Our results document a very accurate out-of-sample performance of our NN-based GA.
%\subsection{Neural Networks}\label{sec_nn}
To learn the relation between input portfolios and the associated GA, we use function approximations in terms of NNs.
Following the presentations in \cite{lutkebohmert2022robust} and \cite{goodfellow2016deep}, we defined a (feed-forward) NN with input dimension $d_{\operatorname{in}}\in \N$, output dimension $d_{\operatorname{out}}\in \N$, $l\in \N$ layers, and non-constant activation function $\varphi:\R\rightarrow \R$ as a function of the form 
\begin{align*}
\R^{d_{\operatorname{in}}}&\rightarrow\R^{d_{\operatorname{out}}}\\
x &\mapsto A_l \circ \varphi \circ A_{l-1} \circ  \dots \circ \varphi\circ A_0(x),
\end{align*}
where $(A_i)_{i=0,\dots,l}$ are affine functions $A_i:\R^{h_i} \rightarrow \R^{h_{i+1}}$, and where the activation function is applied component-wise. The number $h_i \in \N$ is called the \emph{number of neurons} of layer~$i$.
We denote the class of all NNs with input dimension $d_{\operatorname{in}}$, output dimension $d_{\operatorname{out}}$, and with an arbitrary number of layers by $\mathfrak{N}^{\varphi}_{d_{\operatorname{in}},d_{\operatorname{out}}}$. One of the crucial properties of NNs justifying its use for function approximation is its universal approximation property, see, e.g.
\cite[Theorem 2]{hornik1991approximation} or \cite{pinkus1999approximation}.
\begin{lemma}[Universal Approximation Theorem]\label{lem_universal}
Let the activation function $\varphi$ be continuous and non-polynomial.
Then, for any compact set $\mathbb{K} \subset \R^{d_{\operatorname{in}}} $ the set $\mathfrak{N}_{d_{\operatorname{in}},{d_{\operatorname{out}}}}^{\varphi}|_{\mathbb{K}}$ is dense in ${C}(\mathbb{K},\R^{d_{\operatorname{out}}})$ with respect to the topology of uniform convergence on $C(\mathbb{K},\R^{d_{\operatorname{out}}})$.\footnote{We denote by $C(X,Y)$ the space of continuous functions $f:X \rightarrow Y$.}
\end{lemma} 

\subsection{NN-Based GA for Actuarial Approach}

We observe that $GA_q^{\rm exact}(L)$ in (\ref{equ GA exact actuarial}) can be approximated arbitrarily well by NNs whenever we restrict the inputs to a compact set. The inputs consist of the exposure shares $a_n$, the unconditional default probabilities $\operatorname{PD}_n$, the expected loss given defaults $\operatorname{ELGD}_n$, and the sensitivities $\omega_n$ to the systematic risk factor for each obligor $n=1,\ldots,N$. 
Since restricting the inputs $(a_n, \operatorname{PD}_n, \operatorname{ELGD}_n,\omega_n)$ for $n=1,\ldots,N$ to a compact set is natural, the presented approach provides a tractable alternative to other computationally more demanding or less accurate approaches.

\begin{lemma}[Universal Approximation]\label{lem_univ_approx_GA_act}
Assume $\xi >0$ and $\nu>0$. Then, for all $q\in (0,1)$, for all $\varepsilon>0$, ${N} \in \N$, and all compact sets $\mathbb{K} \subset \R^{4{N}}$, there exists a fully connected feedforward neural network $\mathcal{N}\mathcal{N} \in \mathfrak{N}^{\varphi}_{4N,1}$, with activation function $\varphi(x) = \max\{x,0\}$ such that for all $(a, \operatorname{PD}, \operatorname{ELGD},\omega) \in \mathbb{K}$ we have
\begin{equation}\label{eq_eps_expression}
\Bigg|\mathcal{N\mathcal{N}}  \left(a, \operatorname{PD},\operatorname{ELGD},\omega\right)-GA_q^{\rm exact,act}(L) \Bigg|<\varepsilon.
\end{equation}
\end{lemma}

\begin{proof}
Let $q\in (0,1)$, $\varepsilon>0$, ${N} \in \N$, and consider some compact set $\mathbb{K} \subset \R^{4{N}}$. 
Let $Y:=(a,\PD,\ELGD,\omega) =(a_n,\PD_n,\ELGD_n,\omega_n)_{n=1,\dots,N} \in \mathbb{K}$. Then, we define the (random) loss depending on $Y$ by 
$
L(Y):=\sum_{n=1}^N a_n \LGD_n  D_n,
$
where $D_n$ and $\LGD_n$ are as defined in Section~\ref{sec actuarual}.
According to the decomposition outlined in \eqref{equ GA exact actuarial}, we have
{\begin{small}
\begin{equation}\label{eq_GA_proof_1}
\begin{aligned}
GA^{\rm exact,act}_q(L(Y))&= \alpha_q\left(L(Y)\right)- \sum_{n=1}^N a_n \ELGD_n \PD_n (1+\omega_n(\alpha_q(X)-1)),
\end{aligned}
\end{equation}
\end{small}}
where $X$ is Gamma-distributed with mean 1 and variance $1/\xi$ for some $\xi>0$.
Hence, the second summand of \eqref{eq_GA_proof_1} depends continuously on $Y$. To see that the first summand $\alpha_q\left(L(Y)\right)$ also depends continuously on $Y$, we aim at applying Lemma~\ref{lem_cont_distrib}. To that end, we denote by $f_X$ the density function of $X$ and compute for $\ell \in \R$:
{\begin{footnotesize}
\begin{align}
&\mathbb{P}\left(L(Y) \leq \ell \right)\notag \\
 = &\sum_{(d_1,\dots,d_N)\in \{0,1\}^N}  \mathbb{P}\left(L(Y) \leq \ell ~|~ D_1=d_1,\dots,D_N=d_N \right) \cdot \mathbb{P}(D_1=d_1,\dots,D_N=d_N) \notag\\
 &= \sum_{(d_1,\dots,d_N)\in \{0,1\}^N}  \bigg( \mathbb{P}\left(L(Y) \leq \ell ~|~ D_1=d_1,\dots,D_N=d_N \right) \notag \\
 &\hspace{3cm}\times \int_{\R} \prod_{n=1}^N \mathbb{P} \left(D_n = d_n~\middle|~X=x\right) f_X(x) dx \bigg) \notag \\
  &= \sum_{(d_1,\dots,d_N)\in \{0,1\}^N}  \bigg(  \mathbb{P}\left(L(Y) \leq \ell ~|~ D_1=d_1,\dots,D_N=d_N \right) \label{eq_L_Y_y_act}\\
 &\hspace{1cm}\times \int_{\R} \prod_{n=1}^N \left(\PD_n(1+\omega_n)(x-1)\right)^{d_n} \left(1-\PD_n(1+\omega_n(x-1)\right)^{1-d_n} f_X(x) dx\bigg). \notag
\end{align}
\end{footnotesize}}

Then, note that, since $\nu>0$ (i.e. $\LGD_i$ is non-constant) we have that $$\mathbb{K}\times \R \ni (Y, \ell) \mapsto \mathbb{P}\left(L(Y) \leq \ell ~|~ D_1=d_1,\dots,D_N=d_n \right) $$ is continuous as the cdf of a scaled sum of independent Beta distributions. Hence,  the representation in \eqref{eq_L_Y_y_act} implies that
 $\mathbb{K}\times \R \ni (Y,\ell) \mapsto \mathbb{P}\left(L(Y) \leq \ell \right)$ is continuous. Moreover, $P(L(Y)\leq \ell)$ is strictly increasing in $\ell$. Next, note that 
$$\overline{\ell}:=\inf_{\ell \in \R}\left\{\mathbb{P}\left(L(Y) \leq \ell \right) =1 \text{ for all } Y \in \mathbb{K} \right\} <\infty,$$
since the loss $L(Y)$ is bounded if we restrict $Y$ to $\mathbb{K}$. Moreover, the loss $L(Y)$ is by definition non-negative. Thus, we have that the quantile $\alpha_q(L(Y)) = \inf \left\{ \ell \in [0, \overline{\ell}]~ \middle|~ \mathbb{P}\left(L(Y) \leq \ell \right) \geq q \right\}$. By Lemma~\ref{lem_cont_distrib}  we obtain then that  $\mathbb{K} \ni Y \mapsto \inf \left\{ \ell \in  [0, \overline{\ell}]~ \middle|~ \mathbb{P}\left(L(Y)\leq \ell \right) \geq q \right\}  =\alpha_q(L(Y)) $ is continuous. In particular, by \eqref{eq_GA_proof_1}, the function $\mathbb{K} \ni  Y \mapsto GA^{\rm exact,act}_q(L(Y))$ is continuous. Thus, we can apply Lemma~\ref{lem_universal} and obtain the existence of a neural network  $\mathcal{N}\mathcal{N} : \R^{4{N}} \rightarrow \R $ such that for  all $Y \in \mathbb{K}$ we have
$
|\mathcal{N\mathcal{N}}  \left(Y\right)-GA_q^{\rm exact,act}(L(Y)) |< \varepsilon.
$
\end{proof}

Motivated by Lemma~\ref{lem_univ_approx_GA_act}, we formulate  Algorithm~\ref{algo_actuarial} in \ref{app algorithms} allowing to train NNs which are, after training, able to accurately approximate GAs for arbitrary inputs. To facilitate learning, we also include the first-order approximation of the GA (see \eqref{equ GA 1st order actuarial}) as an additional input to the NN.\footnote{Note that increasing the input vector $(a,\PD,\ELGD,\omega)$ by an additional input still allows to apply Lemma~\ref{lem_univ_approx_GA_act} since the target value  $GA_q^{\rm exact, act}(L)$ depends continuously on the input. }

\subsubsection{Training}\label{sec training actuarial}

We train, according to Algorithm~\ref{algo_actuarial}, a NN with the following specifications. The variance parameter of the Gamma distributed systematic risk factor $X$ equals $\xi = 0.25$, the number of simulations $N_{\operatorname{Sim}} = 1~000~000$, the number of iterations $N_{\operatorname{Iter}} = 10~000$.
We consider a feedforward NN with $5$ layers and $512$ neurons in each layer and with ReLu activation functions in each layer. An additional dropout layer is added after each hidden layer, where a dropout rate of 20\% is considered. Our choice is based on a hyperparameter study with $5$-fold cross validation. The (unreported) results show that our selection of hyperparameters is among the best performing combinations.

The choice of the distributions from which the input parameters are sampled for the training of the NN depends of course on the concrete application. In Section \ref{sec real data} we apply our NN-based GA to portfolios of MDBs. Therefore, we train the NN using parameter distributions that are representative for a broad range of MDB portfolios. To this end, we build on recent work in \cite{LuetkebohmertSesterShen2023} who construct realistic MDB portfolios based on publicly available data in MDBs' financial statements. The resulting portfolios all have less than 100 borrowers. Therefore, to train our NN, we set $\overline{N}=100$ and sample in each iteration
$
N_{\operatorname{Obligors}} \sim \mathcal{U}(\{10,\dots,100\}).
$
Further, for each $n = 1,\dots,\overline{N}$, we sample 
\begin{align*}
\operatorname{PD}_n \in \{&0.0,0.01, 0.0002,0.0004,0.0006,0.0011,\\
&0.0018,0.004,0.009,0.0146,
 0.0238,0.0759,0.5147\},
 \end{align*}
 which are the default probabilities of the S\&P transition matrix (compare Sec. \ref{sec real data}). We sample these $\operatorname{PD}$s
with respective probabilities $0.00049$, $0.02297$, $0.00881$, $0.02627$, $0.06454$, $0.05865$,
 $0.06928$, $0.03111$, $0.11070$,
 $0.07672$, $0.19922$, $0.10282$, $0.22834,$
which are based on the empirical distribution of obligors in MDB portfolios; compare \cite{LuetkebohmertSesterShen2023}.
The exposure distribution in each MDB portfolio can be approximated by an exponential distribution. Therefore, we sample the exposures in our training portfolios from exponential distributions with different parameters, i.e.,  
\begin{equation}\label{eq_distribution_inputs}
\begin{aligned}
\operatorname{EAD}_n \sim \operatorname{Exp}(\theta) \text{ with } \theta \sim \mathcal{U}([4,30]),~~~a_n:= \operatorname{EAD}_n / \sum_{j=1}^{\overline{N}}
  \operatorname{EAD}_j,
\end{aligned}
\end{equation}
Moreover, we sample the factor weights according to $\omega_n \sim \mathcal{U}([0,1])$. 
Since we do not have any information on the borrower specific LGDs, in each iteration we set $\ELGD=45\%$ or $\ELGD=10\%$ for training the network. The volatility parameter $\nu$ for the LGDs is set to $0.25$ in both cases.
The constant-size $400$-dimensional input vector $(a, \operatorname{PD},\operatorname{ELGD},\omega)$ to the NN is then constructed by \emph{zero-padding}, i.e., all entries $\operatorname{ELGD}_n ,a_n ,\PD_n,\omega_n$ are $0$  for $n = N_{\operatorname{Obligors}}+1,\dots,100$. The implementation code and training dataset are available on GitHub, see \url{https://github.com/juliansester/DL_Concentration_Risk}.

\subsubsection{Accuracy of the Neural Network}

To test the approximation accuracy of the proposed NN approach, we compare the output of the NN-based GA with the first order approximation from \eqref{equ GA 1st order actuarial} as well as with the outcome of MC simulations with $1~000~000$ simulations. To this end, we sample $1000$ input vectors according to the procedure described in Section \ref{sec training actuarial}.
For each of the inputs we compute the GA according to each of the three approaches. Note that in the asymptotic GA, we use the same factor loadings $\omega_n$ for the calculation of the UL capital $\mathcal{K}_n$ as in the two other approaches, i.e., we implicitly set the factor loadings such that the UL capital in the CreditRisk$^+$ setting and the IRB approach are identical. The confidence level is set to $q=0.999$.
Our results are summarized in Table \ref{tbl_results_GA_NN_mtm_0999} (columns 1 and 2) and Figure \ref{fig_exa_approx_actuarial_0999} and indicate a much smaller prediction error for the NN-based GA than for the first order GA approximation with mean error being about 9 times larger for the latter. 
%In Table~\ref{tbl_results_GA_NN_actuarial_2_099}, we only consider those portfolios with less than $25$ obligors. Here the prediction error is more than 11 times larger for $GA^{\rm 1st,act}$ than for $GA^{\rm NN,act}$. 
This documents 
that the NN approach clearly outperforms the approximation from \eqref{equ GA 1st order actuarial} in terms of approximation error and shows the importance of carefully selecting the correct method to compute GAs in small portfolios.

\FloatBarrier
\begin{table}[h!]
\begin{center}{
\scalebox{0.75}{
\begin{tabular}{lrrrr}
\toprule
 & $| GA^{\rm NN,act} -GA^{\rm MC,act}|$ & $| GA^{\rm 1st,act} -GA^{\rm MC,act}|$ & $| GA^{\rm NN,MtM} -GA^{\rm MC,MtM}|$ & $| GA^{\rm 1st,MtM} -GA^{\rm MC,MtM}|$\\
\midrule
%No. of Portf. &                 1000 &                       1000 & 1000 &                       1000 \\
mean  &                    0.00565 &                          0.04897 &  0.00573 &                          0.01084 \\
std   &                    0.00756 &                          0.06346 &0.00586 &                          0.00917 \\
min   &                    0.00000 &                          0.00020 &0.00003 &                          0.00004 \\
25\%   &                    0.00162 &                          0.01696 & 0.00201 &                          0.00397 \\
50\%   &                    0.00351 &                          0.02814 &0.00416 &                          0.00892 \\
75\%   &                    0.00658 &                          0.05463 & 0.00753 &                          0.01509 \\
max   &                    0.08046 &                          0.82322 & 0.06398 &                          0.08030 \\
\bottomrule
% & $| GA^{\rm NN,MtM} -GA^{\rm MC,MtM}|$ & $| GA^{\rm 1st,MtM} -GA^{\rm MC,MtM}|$\\
%\midrule
%No. of Portfolios &                 1000 &                       1000 \\
%mean  &                    0.00573 &                          0.01084 \\
%std   &                    0.00586 &                          0.00917 \\
%min   &                    0.00003 &                          0.00004 \\
%25\%   &                    0.00201 &                          0.00397 \\
%50\%   &                    0.00416 &                          0.00892 \\
%75\%   &                    0.00753 &                          0.01509 \\
%max   &                    0.06398 &                          0.08030 \\
%\bottomrule
\end{tabular}}

}
\end{center}
\caption{The table shows the approximation error of the NN-based $GA^{\rm NN,act}$ (column~1) and $GA^{\rm NN,MtM}$ (column 3) as well as the first order approximation $GA^{\rm 1st,act}$ (column~2) from \eqref{equ GA 1st order actuarial} and $GA^{\rm 1st,MtM}$ (column 4) from \eqref{eq_GA_formula_1st_order_mtm} when compared to $ GA^{\rm MC, act}$ and $ GA^{\rm MC, MtM}$, resp., computed by MC simulation with $10^6$ simulations of the underlying risk factor. We tested the approaches on $1000$ input portfolios that were sampled according to the distributions indicated in Section \ref{sec deep learning}. The confidence level is $q=0.999$.}\label{tbl_results_GA_NN_mtm_0999}
\end{table}
\FloatBarrier

\begin{figure}[h!]
\begin{center}
	\subfigure[First order approximation GA]{
    \includegraphics[width=0.45\textwidth]{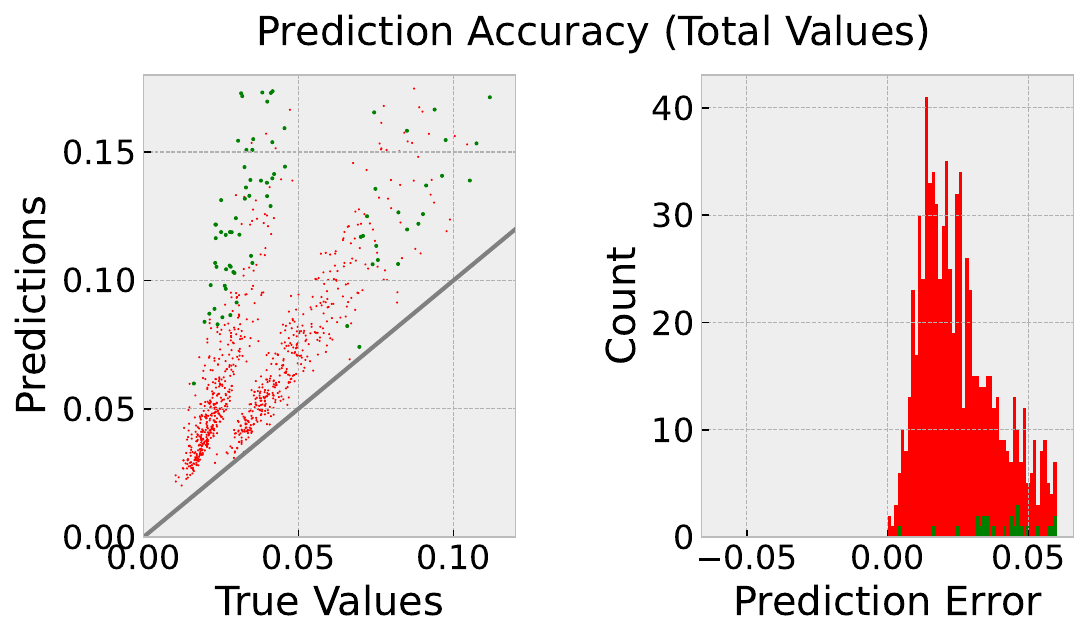} \label{actuarial_approx}}
    	\subfigure[NN-based GA]{
    \includegraphics[width=0.45\textwidth]{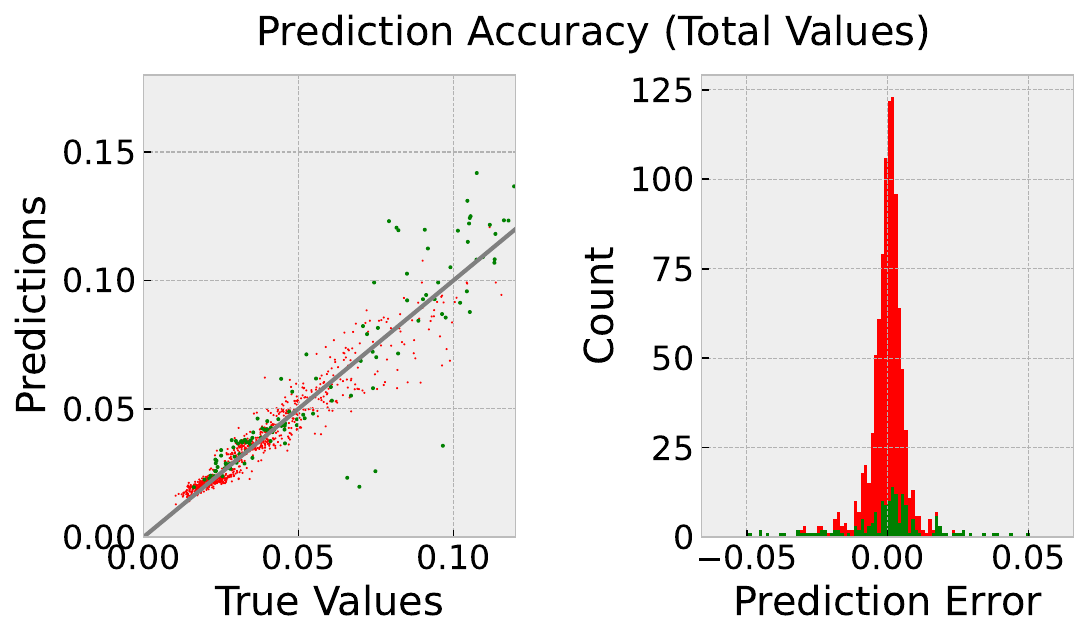} \label{actuarial_NN}}
              \caption{Actuarial approach: The figure illustrates the accuracy of the predictions on the test set of 1000 sampled portfolios. The
left panels of both subfigures (a) and (b) show a plot of all target values (x-values) generated by MC simulation and its predictions (y-values) either according to the first-order GA approximation from \eqref{equ GA 1st order actuarial} (a) or according to the trained NN (b). The right panels depict histograms of the corresponding prediction error, i.e., the error between
target values and predicted values. The green dots correspond to portfolios with less than $25$ borrowers. The confidence level is $q=0.999$.} \label{fig_exa_approx_actuarial_0999}
\end{center}
\end{figure}

\subsection{NN-Based GA for MtM Approach}

When considering the MtM approach, the portfolio is characterized by the exposure shares $a_n$, the expected loss given defaults $\ELGD_n$, the asset correlations $\rho_n$, the coupons $c_n$, the maturities $\tau_n$, and the current rating $g_n$ of each obligor $n=1,\ldots,N$. To compute the GA, we derive the following universal approximation result to approximate the GA from \eqref{equ GA exact MtM} by NNs.

\begin{lemma}[Universal Approximation]\label{lem_univ_approx_GA_mtm}
Let $\mathbb{G} \subset \N$ and $\mathbb{T} \subset \N$  be finite sets.
For all $q\in (0,1)$, for all $\varepsilon>0$, $\nu>0$, number of obligors ${N} \in \N$, number of states $S \in \mathbb{N}$, transition probabilities $(p_{ss'})_{s =1,\dots,S \atop s' = 0,\dots,S}$ and compact sets $\mathbb{K} \subset \R^{4\overline{N}}$, there exists a fully connected feedforward neural network $\mathcal{N}\mathcal{N} \in \mathfrak{N}^{\varphi}_{6N,1}$, with $\varphi(x) = \max\{x,0\}$  such that for all $(a,\operatorname{ELGD},\rho,c,g,\tau) \in \mathbb{K}\times \mathbb{G}^N \times \mathbb{T}^N$ we have
\begin{equation}\label{eq_eps_expression}
\Bigg|\mathcal{N\mathcal{N}}  \left((a,\operatorname{ELGD},\rho,c,g,\tau)\right)-GA_q^{\rm exact, MtM}(L) \Bigg|<\varepsilon.
\end{equation}
\end{lemma}

\begin{proof}
Let $q\in (0,1)$, $\varepsilon>0$, ${N} \in \N$, and consider some compact set $\mathbb{K} \subset \R^{4{N}}$.  We pick some $(g,\tau)\in  \mathbb{G}^N \times \mathbb{T}^N$ and some $Y:=(a,\operatorname{ELGD},\rho,c) \in  \R^{4{N}}$.
Then, we define the loss in dependence of $Y$, $g$, and $\tau$ by\footnote{Note that the loss is defined as $(\mathbb{E}[R]-R)e^{-rT}$. However, as mentioned before equation~\eqref{equ GA exact MtM}, the expected return $\mathbb{E}[R]$ cancels out when computing the GA. Hence, we replace without loss of generality the loss by the negative return $-R$ when computing the GA.  }
\[
L(Y,g,\tau) := -\sum_{n=1}^N a_n \sum_{s=1}^S \frac{P_{nT}(s)}{P_{n0}}\cdot\one_{\left\{\Phi^{-1}\left(\sum_{i=0}^{s-1} p_{g(n),i}\right)< ~Y_n~\leq \Phi^{-1}\left(\sum_{i=0}^{s} p_{g(n),i}\right)\right\}},
\]
where the transition probabilities $p_{g(n),i}$ only depend on the rating $g(n)$ of each obligor $n$.
According to \eqref{equ GA exact MtM} we have
\begin{equation}\label{eq_GA_MtM_approx}
    GA^{\rm exact, MtM}_q(L(Y,g,\tau))= \alpha_q(L(Y,g,\tau))- \E[L(Y,g,\tau)~|~\alpha_{1-q}(X)] ,
\end{equation}
where it can be easily seen that the second summand of \eqref{eq_GA_MtM_approx}
\[
\begin{array}{ccl}
\E[L(Y,g,\tau)~|~\alpha_{1-q}(X)] &=&\ds e^{-rT} \sum_{n=1}^N a_n \cdot \mu_n(\alpha_{1-q}(X))\\[2mm]
&=&\ds  e^{-rT} \sum_{n=1}^N a_n \cdot \sum_{s=0}^S \mathbb{E}[R_n|S_n=s] \pi_{ns}(\alpha_{1-q}(X)) 
\end{array}
\]
depends continuously on $Y=(a,\operatorname{ELGD},\rho,c)$. Here, $R_n$ is defined by (\ref{equ return R_n}) and $\pi_{ns}(X)$ is given by (\ref{eq_pi_def}). We denote by $f_X$ the density of a standard normal distribution and obtain for all $\ell \in \R$ that
\begin{align*}
\mathbb{P}\left(L(Y,g,\tau) \leq l \right) & = \sum_{(s_1,\dots,s_N) \in \{1,\dots,S\}^N} \bigg(\mathbb{P}\left(L(Y,g,\tau) \leq \ell~\middle|~S_1=s_1,\dots,S_N = s_N\right)\\
&\hspace{4cm} \times \mathbb{P}\left(S_1=s_1,\dots,S_N = s_N\right)  \bigg) \\
& = \sum_{(s_1,\dots,s_N) \in \{1,\dots,S\}^N} \bigg( \mathbb{P}\left(L(Y,g,\tau) \leq \ell~\middle|~S_1=s_1,\dots,S_N = s_N\right)\\
&\hspace{4cm} \times \int_{\mathbb{R}} \prod_{n=1}^N \pi_{ns}(x) f_X(x) dx \bigg)
\end{align*}
which depends continuously on $Y$ and $\ell$ as $\nu>0$, and is strictly increasing in $\ell$. Note further that 
$$\overline{\ell}:=\inf_{\ell \in \R}\left\{\mathbb{P}\left(L(Y,g,\tau)  \leq \ell \right) =1 \text{ for all } Y \in \mathbb{K} \right\} <\infty,$$
since the loss $L(Y,g,\tau) $ is bounded if we restrict $Y$ to $\mathbb{K}$. Moreover,  $-L(Y,g,\tau) $ is by definition non-negative. Thus, we have 
\begin{align}
\alpha_q(L(Y,g,\tau) )&=-\alpha_{1-q}(-L(Y,g,\tau) )\nonumber  \\
&= -\inf \left\{ \ell \in [0, \overline{\ell}]~ \middle|~ \mathbb{P}\left(-L(Y,g,\tau) \leq \ell \right) \geq 1-q \right\}\nonumber
\end{align}
for some $\overline{\ell} >0$. Due to Lemma~\ref{lem_cont_distrib}, we then obtain that  $\mathbb{K} \ni Y \mapsto \alpha_q(L(Y,g,\tau) ) $ is continuous and thus, by \eqref{eq_GA_MtM_approx}, also $\mathbb{K} \ni  Y \mapsto GA^{\rm exact, MtM}_q(L(Y,g,\tau) )$ is continuous. 
Applying linear interpolation on the grid $ \mathbb{G}\times \mathbb{T}$, we may then conclude the existence of a continuous function 
$\R^{4N} \times \R \times \R \ni (Y,g,\tau) \mapsto \Psi(Y,g,\tau)$ such that $\Psi(Y,g,\tau)|_{\mathbb{K}\times \mathbb{G}\times \mathbb{T}} = GA^{\rm exact, MtM}_q(L(Y,g,\tau) ).$

Eventually, we apply Lemma~\ref{lem_universal} to $\Psi$ and obtain the existence of a neural network  $\mathcal{N}\mathcal{N} : \R^{6{N}} \rightarrow \R $ such that for  all $(Y,g,\tau) \in \mathbb{K}\times \mathbb{G}\times \mathbb{T}$ we have 
\begin{equation}\nonumber%\label{eq_nn_2}
\Bigg|\mathcal{N}\mathcal{N} \left(Y,g,\tau\right)-\Psi(L(Y,g,\tau) ) \Bigg|=\Bigg|\mathcal{N}\mathcal{N} \left(Y,g,\tau\right)-GA^{\rm exact, MtM}_q(L(Y,g,\tau) )  \Bigg|< \varepsilon.
\end{equation}

%Note that conditional on the realization of $\LGD$ we have
%\begin{align*}
%\mathbb{P}\left(R_n = \frac{P_{nT}(s)}{P_{n0}}~\middle|~\LGD \right) &= \sum_{\tilde{s} =0}^S \mathbb{P}\left(R_n= \frac{P_{nT}(s)}{P_{n0}}~\middle|~S_n = \tilde{s} , \LGD \right) \mathbb{P}\left(S_n = \tilde{s}~\middle|~ \LGD\right) \\
% &= \sum_{\tilde{s} =0}^S  \one_{\{\tilde{s}=s\}} \mathbb{P}(S_n = \tilde{s}| \LGD) = \mathbb{P}(S_n = s~|~\LGD).
%\end{align*}
%%where $p_{ns}$ is the unconditional probability that borrower $n$ is in rating grade $s$ at time $T$.
%Hence $R_n$ is, conditional on $\LGD$, a discrete random variable attaining values only in the set $\{\frac{P_{nT}(0)}{P_{n0}},\dots,\frac{P_{nT}(s)}{P_{n0}}\}$. Then, with the same arguments as in the proof of Lemma~\ref{lem_univ_approx_GA_act} we can, for all $\delta>0$, approximate $R_n$ by $R_n^{\delta}$ which has a continuous cumulative distribution function. Then, by approximating the granularity adjustment of the loss w.r.t.\,the random variables $R_n^{\delta}$ by neural networks, we obtain the approximation result.
\end{proof}

In the MtM setting, Algorithm~\ref{algo_MtM} in \ref{app algorithms} allows to train NNs which are then able to compute GAs for any input of portfolio parameters. To improve learning and in line with the approach for the actuarial case, we include the first order approximation $GA^{\rm 1st, MtM}$ in \eqref{eq_GA_formula_1st_order_mtm} as an additional input to the NN.\footnote{Note that Lemma~\ref{lem_univ_approx_GA_act} can still be applied after increasing the input vector $(a,\operatorname{ELGD},\rho,c,g,\tau)$ by this additional input since the target value $GA_q^{\rm exact, MtM}(L)$ depends continuously on the input. }

\subsubsection{Training}\label{sec training mtm}

For training the NN, according to Algorithm~\ref{algo_MtM}, we consider similar specifications as for the actuarial case, i.e. the number of simulations $N_{\operatorname{Sim}} = 1~000~000$, the number of iterations $N_{\operatorname{Iter}} = 10~000$. We consider a feedforward NN with $5$ layers and $512$ neurons in each layer, where a ReLu activation function is applied to  each layer, and no dropout or regularization. Our choice is based on a hyperparameter study with $5$-fold cross validation. The (unreported) results show that our selection of hyperparameters is among the best performing combinations.

To train the NN, we set $\overline{N}=100$ and we sample in each iteration
$
N_{\operatorname{Obligors}} \sim \mathcal{U}(\{10,\dots,100\}).
$
Similarly to Section \ref{sec training actuarial}, for $n = 1,\dots,\overline{N}$ we sample the exposures from an exponential distribution. Further parameters are sample according to uniform distributions, i.e.,
\begin{equation}\label{eq_distribution_inputs_mtm}
\begin{aligned}
%\operatorname{ELGD}_n &\sim  \mathcal{U}([0,1]), \\
\operatorname{EAD}_n &\sim \operatorname{Exp}(\theta) \text{ with } \theta \sim \mathcal{U}([4,30]),~~a_n:= \operatorname{EAD}_n / \sum_{j=1}^{\overline{N}}
  \operatorname{EAD}_j, \\
&\rho_n \sim \mathcal{U}([0.15,0.7]),~~c_n \sim   \mathcal{U}([0,0.1]),
~~\D_n \sim  \mathcal{U}([1,10]),
\end{aligned}
\end{equation}
The range of the asset correlation parameters stems from the empirical study in \cite{LuetkebohmertSesterShen2023}. Since there is no publicly available information on coupon rates, we sample these from a rather wide range of values. According to \cite{LuetkebohmertSesterShen2023} maturities of sovereign loans range between a few months and up to 35 years with the majority of loans maturing within the next 10 years, and are mostly uniformly distributed over this time period. Therefore, to facilitate training, we assume maturities are uniformly distributed between one and ten years.
We sample the initial ratings $g_n$ (analogous to the PDs in the actuarial case) from the probabilities implied by the S\&P transition matrix (compare Sec. \ref{sec real data}) and the overall empirical distribution of obligors in MDB portfolios described in \cite{LuetkebohmertSesterShen2023}. The rating transition matrix is also used for training the NN approach in the MtM setting. Thus, we consider 17 non-default states and an absorbing default state.
Similarly as before, we set $\ELGD=45\%$ or $\ELGD=10\%$ and the volatility parameter $\nu= 0.25$ for training the network.
% {\textcolor{red}{The correlation parameter could also be set to a fixed value of 0.35 or we should assume $\rho_n=\rho$ for all $n$ and then sample $\rho$ uniformly distributed on [0,1]. Further, the durations and coupons could also be constant across obligors so that we simply work with the average maturity and weighted average spread mentioned in the financial statements.}}
The constant-size $600$-dimensional input vector $(a,\operatorname{ELGD},\rho,c,g,D)$ to the NN is then constructed by \emph{zero-padding}, i.e., all entries $a_n,\operatorname{ELGD}_n, \rho_n,c_n, g_n, D_n$ are $0$  for $n = N_{\operatorname{Obligors}}+1,\dots,100$. The implementation code and training dataset are available on GitHub, see \url{https://github.com/juliansester/DL_Concentration_Risk}.

%{\textcolor{red}{How do we specify $N_{avg}$?}}\\

\subsubsection{Accuracy of the Neural Network}

We compare the output of the NN-based GA and the first order approximation from \eqref{eq_GA_formula_1st_order_mtm} with GA computations based on MC simulations with $1~000~000$ simulations. To this end, we consider $1000$ different input vectors, which are sampled according to the procedure outlined in Section \ref{sec training mtm}.
For each of the inputs, we compute the GA according to each of the three approaches for the confidence level $q=0.999$. Results are summarized in Table~\ref{tbl_results_GA_NN_mtm_0999} (columns 3 and 4) and Figure~\ref{fig_exa_approx_mtm_0999}. As can be observed, the prediction error for the first order GA approximation from \eqref{eq_GA_formula_1st_order_mtm} is on average about twice as large as for the NN GA. This documents that the NN GA significantly outperforms the analytic GA with respect to the approximation error. 
%When considering only the very small portfolios with less than 25 obligors as shown in Table~\ref{tbl_results_GA_NN_mtm_2_099}, the effect is comparable. 
Overall, the results show that even for high quantiles as $q=0.999$ the NN can be efficiently trained to accurately predict the GA for small and concentrated portfolios.

\begin{figure}[h]
\begin{center}
	\subfigure[First order approximation GA]{
    \includegraphics[width=0.45\textwidth]{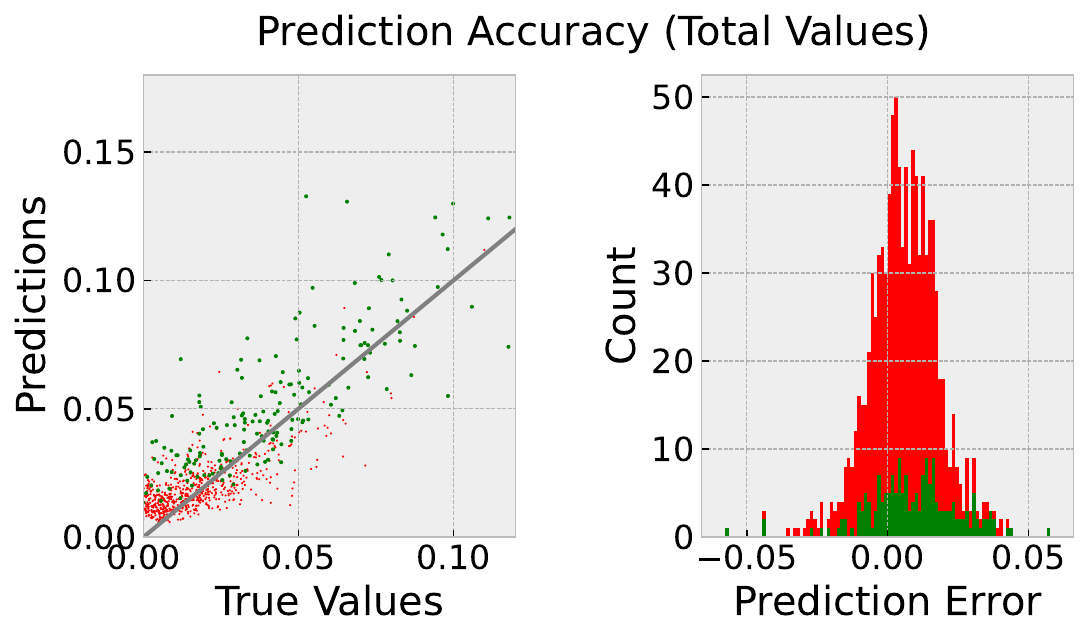} \label{actuarial_approx}}
    	\subfigure[NN-based GA]{
    \includegraphics[width=0.45\textwidth]{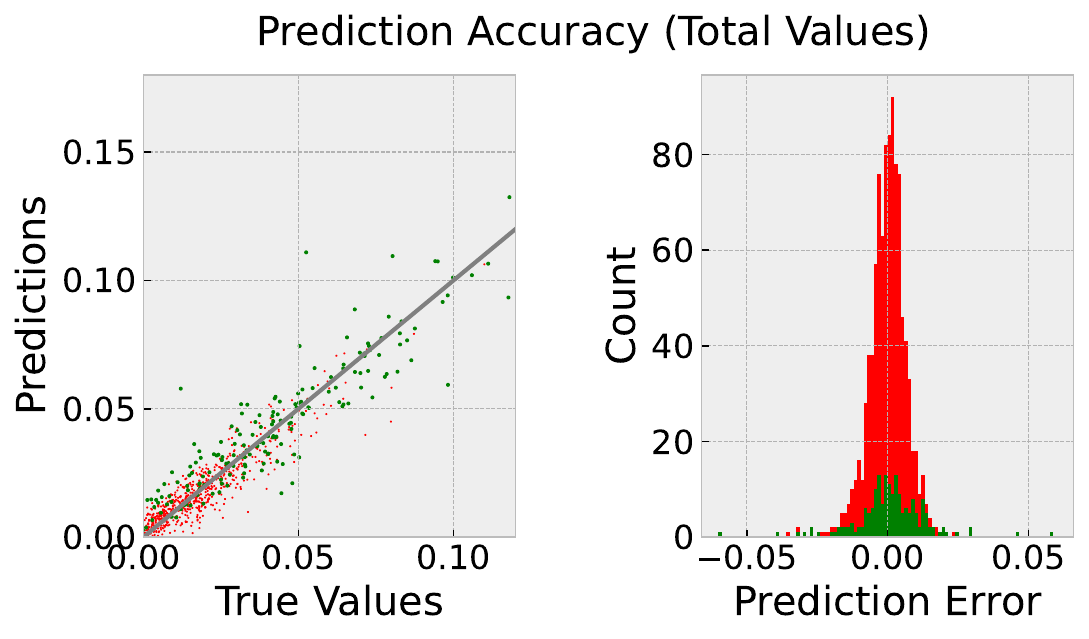} \label{actuarial_NN}}
              \caption{MtM approach: The figure illustrates the accuracy of the predictions on the test set of 1000 sampled portfolios. The
left panels of both subfigure (a) and (b) show a plot of all target values (x-values) generated by MC simulation and its predictions (y-values) either according to the first order GA approximation from \eqref{eq_GA_formula_1st_order_mtm} (a) or according to the trained NN (b). The right panels depict histograms of the corresponding prediction error, i.e., the error between
target values and predicted values. The green dots correspond to portfolios with less than $25$ borrowers. The confidence level is $q=0.999$.} \label{fig_exa_approx_mtm_0999}
\end{center}
\end{figure}

%\begin{table}[h]
%\begin{center}{
%\begin{tabular}{lrr}
%\toprule
% & $| GA^{\rm NN,MtM} -GA^{\rm MC,MtM}|$ & $| GA^{\rm 1st,MtM} -GA^{\rm MC,MtM}|$\\
%\midrule
%No. of Portfolios &                     184 &                          184 \\
%mean  &                       0.00882 &                            0.01500 \\
%std   &                       0.00932 &                            0.01310 \\
%min   &                       0.00003 &                            0.00009 \\
%25\%   &                       0.00257 &                            0.00512 \\
%50\%   &                       0.00645 &                            0.01228 \\
%75\%   &                       0.01181 &                            0.02098 \\
%max   &                       0.06398 &                            0.08030 \\
%\bottomrule
%\end{tabular}}
%\end{center}
%\caption{Mtm approach: The table shows the approximation error of the NN-based $GA^{\rm NN,MtM}$ and the first order approximation $GA^{\rm 1st,MtM}$ from \eqref{eq_GA_formula_1st_order_mtm} when compared to $ GA^{\rm MC, MtM}$ computed by MC simulation with $10^6$ simulations of the underlying risk factor. We tested the approaches on $1000$ inputs that were sampled according to the distributions indicated in Section \ref{sec training mtm}, however in this figure, we only consider those portfolios with less than $25$ obligors. The confidence level is $q=0.999$.}\label{tbl_results_GA_NN_mtm_2_099}
%\end{table}

\section{Performance Evaluation of the NN Approach}\label{sec performance evaluation}

\subsection{Sensitivity Analysis based on Simulated Portfolios}\label{sec sensitivity analysis}

In this section, we study the impact of different input parameters on the size of the GA. 
%\subsection{Actuarial Approach}
To this end, we sample a portfolio with $100$ obligors with input parameters distributed as described in Sections \ref{sec training actuarial} and \ref{sec training mtm} for the actuarial approach and the MtM approach, respectively. 
First we investigate the effect of reducing the number of obligors by gradually deleting obligors from the originally sampled portfolio, depicted in Figure~\ref{fig_sensitivities_actuarial}~(a) for the actuarial and in Figure~\ref{fig_sensitivities_mtm}~(a) for the MtM approach. We observe that with an increasing number of obligors the GA tends to decrease, while this relation is not necessarily monotone since adding single obligors with large exposure or high PDs might also increase the GA as can be clearly observed in the MtM case. Figure~\ref{fig_sensitivities_actuarial}~(a) further shows  that the approximate GA in the actuarial approach overestimates the true GA (approximated by the NN) significantly which is also in line with the results in Figure \ref{fig_exa_approx_actuarial_0999}. In the MtM approach, it highly depends on the choice of the portfolio whether the approx. GA over- or underestimates the true GA with overestimation occurring more frequently than underestimations (compare Figure~\ref{fig_sensitivities_mtm}~(a) and also Figure \ref{fig_exa_approx_mtm_0999}).

Next, we study how changes in the credit quality or the relative weights of individual loans in the portfolio affect the GA. Therefore, to consider an average-sized portfolio, we construct a portfolio consisting of only $50$ obligors sampled according to our input distributions for the actuarial and the MtM case. We depict the composition of these portfolios ordered by an increasing exposure share together with the assigned $\PD$ in Figures~\ref{fig_sensitivities_actuarial}~(b) and ~\ref{fig_sensitivities_mtm}~(b).
To study the effect of a change in the credit  quality of an individual borrower on the name concentration risk of the credit portfolio, we consider a single notch downgrade for each of the obligors in the respective portfolios. It can be observed that the effect on the GA is the more pronounced the larger the exposure of the respective obligor is, i.e., if a borrower with a large relative weight gets downgraded, this increases the name concentration of the portfolio and therefore leads to a rising GA. The opposite effect occurs when the rating of the borrower improves (compare Figures~\ref{fig_sensitivities_actuarial}~(c) and ~\ref{fig_sensitivities_mtm}~(c)).\footnote{This is especially interesting in light of the remark in \cite{UbertiFigini2010} that differences in borrower specific default probabilities are more relevant for small portfolios than for larger and more diversified portfolios.}

% We then multiply the borrower-specific PDs in the actuarial case by a multiplier varying between 0.5 and 2.0 (ensuring that PDs stay below one). The results are depicted in right graph in Figure~\ref{fig_sensitivities_actuarial}.
%Similarly, we multiply the borrowers' rating $g$ in the MtM case by a multiplier and then round the resulting rating class $g$ to the next lower integer value. The corresponding sensitivity is shown in the right graph of Figure \ref{fig_sensitivities_mtm}. As expected the GA increases with increasing PD, resp. with decreasing rating class~$g$.

\begin{figure}[h]
\begin{center}
	\subfigure[Sensitivity w.r.t.\,number of obligors.]{
       \includegraphics[width=0.43\textwidth]{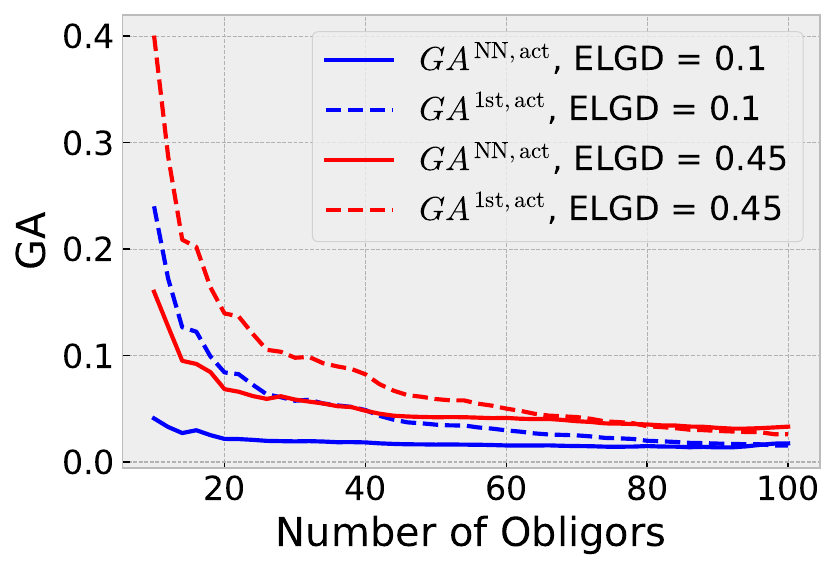} \label{sensitivities_1}}
     	\subfigure[Reduced Portfolio with 50 Obligors]{
    \includegraphics[width=0.49\textwidth]{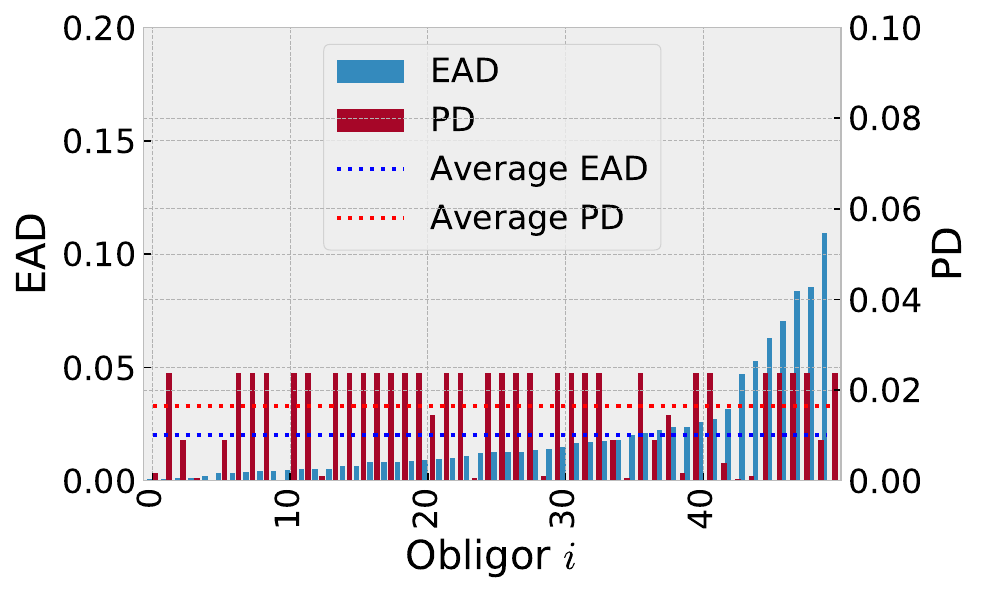}
   \label{sensitivities_2}}
   \subfigure[Sensitivity w.r.t.\,1 notch downgrade]{
       \includegraphics[width=0.43\textwidth]{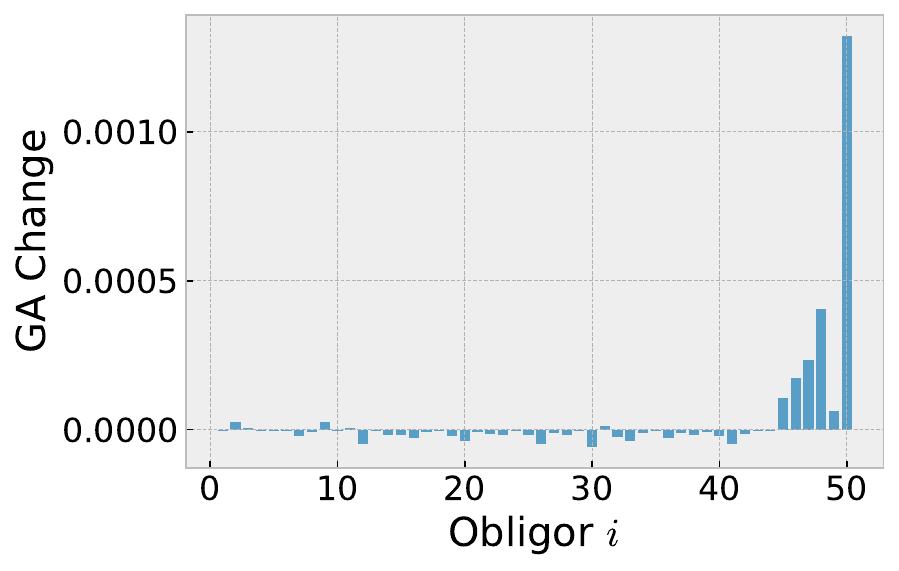} \label{sensitivities_3}}
     	\subfigure[Sensitivity w.r.t.\,1\% change in $a_n$]{
    \includegraphics[width=0.46\textwidth]{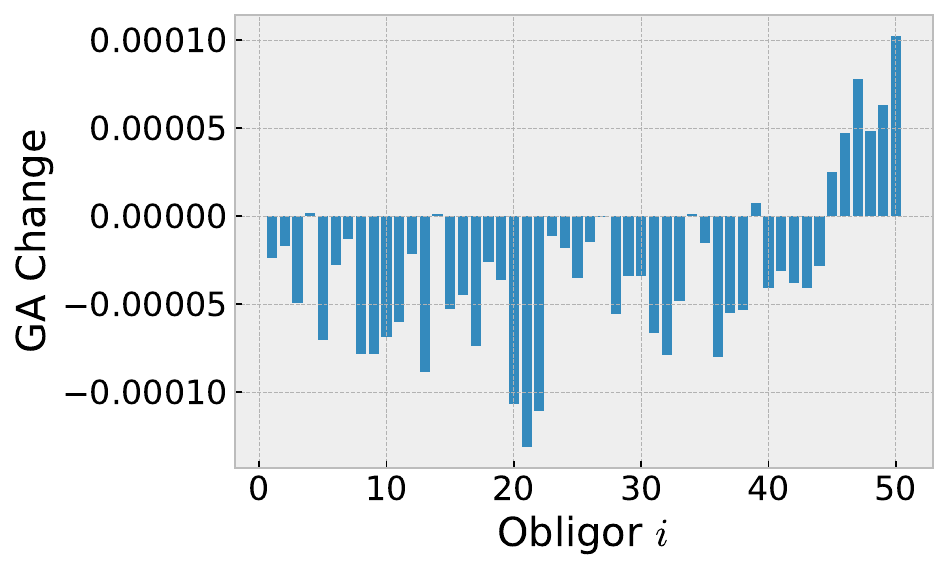}
   \label{sensitivities_4}}
     \caption{The figure shows the sensitivities of the approximate and NN GA in the actuarial approach of a portfolio sampled according to Section \ref{sec training actuarial} w.r.t. the number of obligors (a). In Subfigure (b) we show the composition of a reduced portfolio of 50 obligors, which we consider in (c) 
to study the impact of a single notch downgrade on $GA^{\rm NN,act}$. The impact of an increase in an obligor's exposure share $a_n$ on $GA^{\rm NN,act}$ is depicted in (d).} \label{fig_sensitivities_actuarial}
\end{center}
\end{figure}

Finally, we investigate the influence of increasing the relative portfolio weight (given by $a_i=\operatorname{EAD}_i / \sum_{j=1}^{50}\operatorname{EAD}_j$) of a single obligor by $1\%$.\footnote{After increasing the weight we normalize the portfolio again such that the weights still sum up to $1$.}  The extent to which this change in the portfolio composition affects the GA depends crucially on both the relative portfolio weight of the respective obligor and its PD, whereby the influence of the relative portfolio weight appears to be stronger: If the relative portfolio weight before the increase is rather low, the effect of the $1\%$ increase is usually a reduction in the GA, as the portfolio becomes more homogeneous so that name concentration is reduced. However, if the relative weight is already high, then a further increase in the weight  increases also the name concentration and hence the GA, unless the PD is small. In the latter case, an increase in the relative weight can lead to a better overall quality of the portfolio and thus to a reduction in the GA. Compare for a detailed illustration Figures~\ref{fig_sensitivities_actuarial}~(d) and ~\ref{fig_sensitivities_mtm}~(d).

\begin{figure}[h]
\begin{center}
	\subfigure[Sensitivity w.r.t.\,number of obligors.]{
       \includegraphics[width=0.43\textwidth]{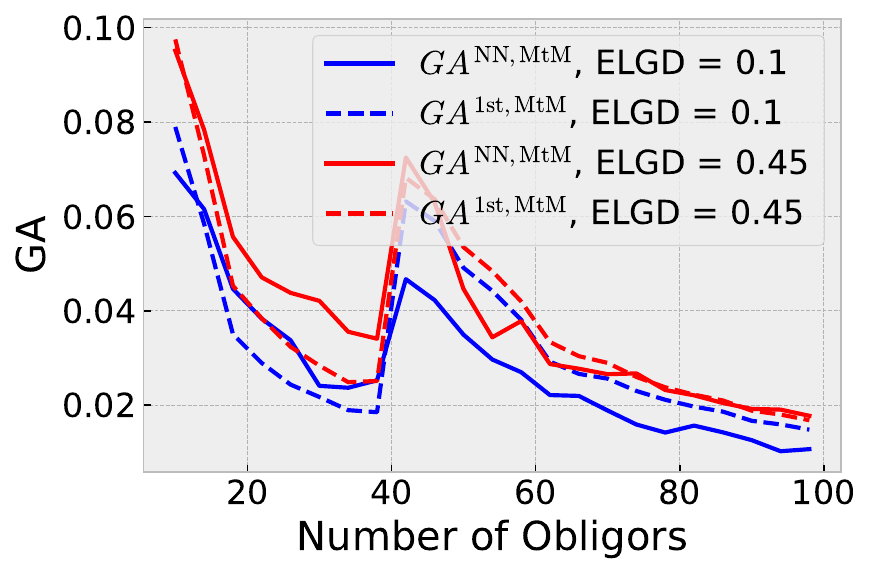} \label{sensitivities_1}}
     	\subfigure[Reduced Portfolio with 50 Obligors]{
    \includegraphics[width=0.49\textwidth]{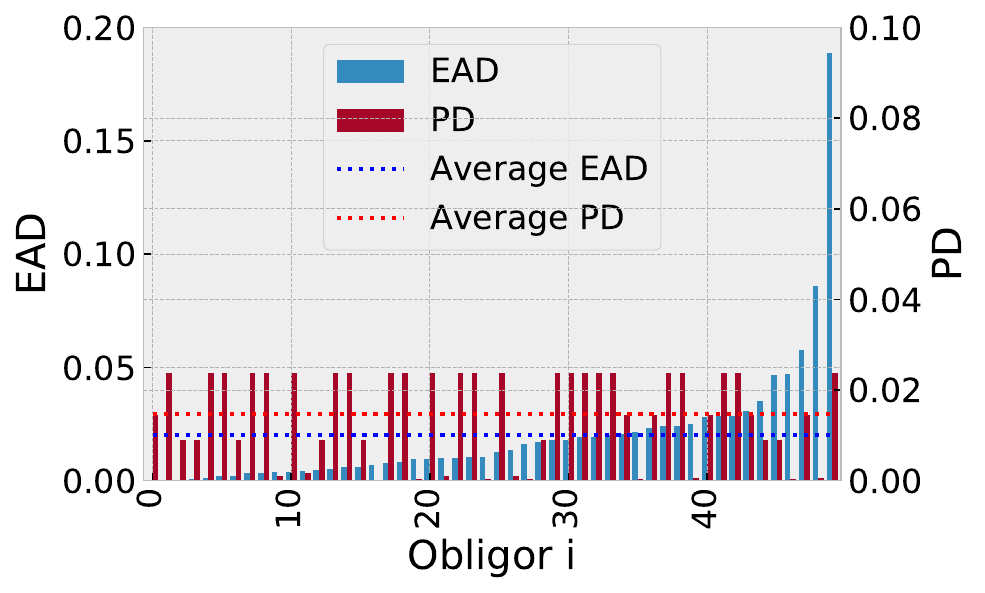}
   \label{sensitivities_2}}
   \subfigure[Sensitivity w.r.t.\,1 notch downgrade]{
       \includegraphics[width=0.43\textwidth]{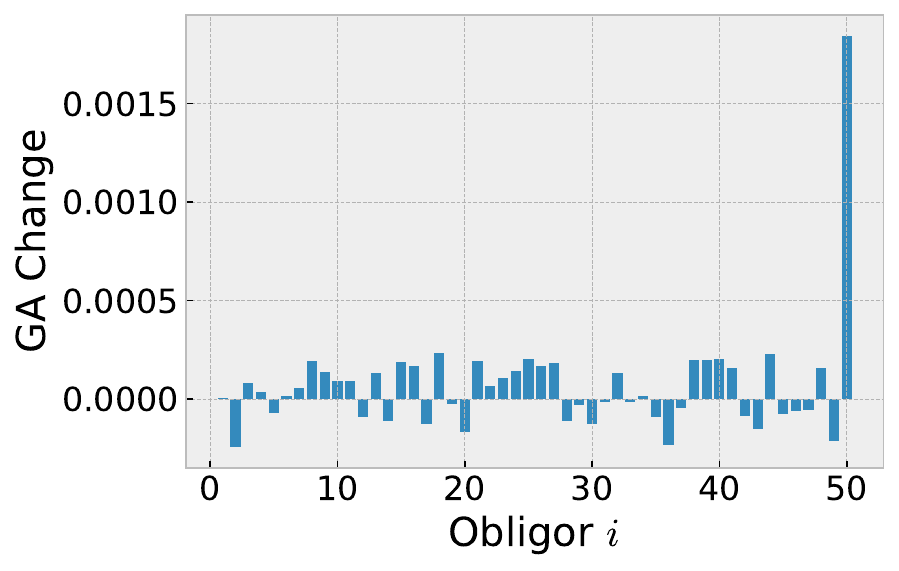} \label{sensitivities_3}}
     	\subfigure[Sensitivity w.r.t.\,1\% change in $a_n$]{
    \includegraphics[width=0.46\textwidth]{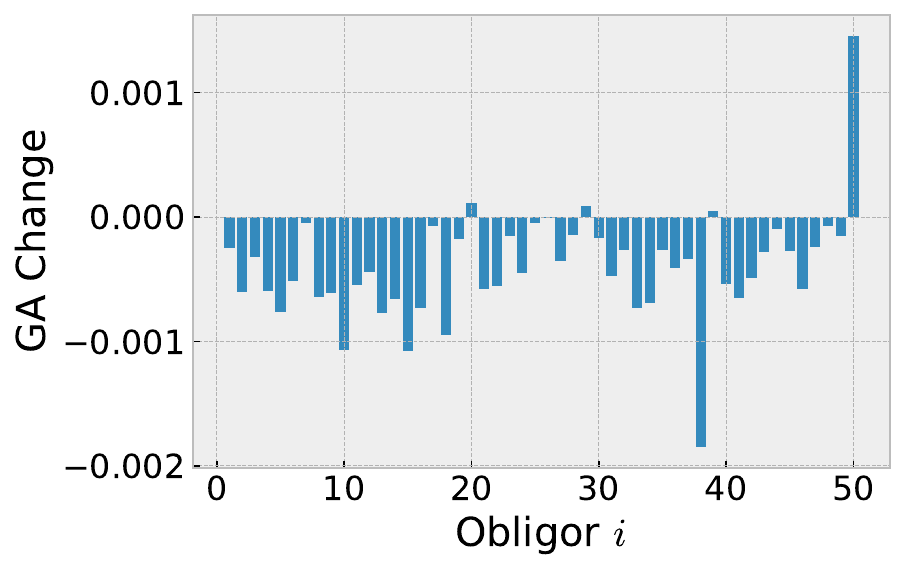}
   \label{sensitivities_4}}
              \caption{The figure shows the sensitivities of the approximate and NN GA in the MtM approach of a portfolio sampled according to Section \ref{sec training mtm} w.r.t. the number of obligors (a). In Subfigure (b) we show the composition of a reduced portfolio of 50 obligors, which we consider in (c) 
to study the impact of a single notch downgrade on $GA^{\rm NN,MtM}$. The impact of an increase in an obligor's exposure share $a_n$ on $GA^{\rm NN,MtM}$ is depicted in (d)} \label{fig_sensitivities_mtm}
\end{center}
\end{figure}

Overall, these results document that the impact of changes in individual borrower's characteristics or exposure shares on the name concentration risk of the loan portfolio is quite complex. The effect of including an additional loan to an existing portfolio is comparable. Thus, quantitative tools as our proposed method are very useful to evaluate these effects. In particular, our NN-based GA allows for a fast and accurate evaluation of such sensitivities whereas comparable calculations based on MC simulations are computationally expensive. In this way, our approach can help, for instance, to examine the impact of a potential future downgrading of a borrower on the name concentration risk of the portfolio, which could be extremely useful for rating agencies, for example. Alternatively, our NN-based GA can be applied to evaluate how adding loans to a portfolio may affect the portfolio's exposure to name concentration risk. This might be very valuable for the investment decisions of institutions if they want to maintain certain risk assessments.

\subsection{Application to Realistic MDB Portfolios}\label{sec real data}

Next, we evaluate the performance of our NN GA based on realistic MDB portfolios as constructed in \cite{LuetkebohmertSesterShen2023} for eleven MDBs from their publicly available financial statements:
{
\begin{itemize}\setlength\itemsep{-0.2cm}
    \item African Development Bank (AfDB); see \cite{AfDB2022}
    \item Asian Development Bank (ADB); see \cite{ADB2022}
    %\item Asian Infrastructure Investment Bank (AIIB)
    %\item Bank of the South
    \item Development Bank of Latin America and the Caribbean (CAF); see \cite{CAF2022}
    \item Caribbean Development Bank (CDB); see \cite{CDB2022}
    \item Central American Bank for Economic Integration (CABEI); see \cite{CABEI2022}
    %\item Development Bank of the Central African States (BDEAC)
    \item East African Development Bank (EADB); see \cite{EADB2022}
    %\item Eurasian Development Bank (EDB)
    \item European Bank for Reconstruction and Development (EBRD); see \cite{EBRD2022}
    %\item European Investment Bank (EIB)
    \item Inter-American Development Bank (IDB); see \cite{IDB2022} %sometimes IADB
    %\item International Development Finance Club (IDFC)
    \item International Bank for Reconstruction and Development (IBRD); see \cite{IBRD2022}
    %\item International Development Association (IDA)
    %\item Islamic Development Bank (IsDB)
    %\item New Development Bank (NDB)
    %\item Nordic Investment Bank (NIB)
    \item Trade and Development Bank (TDB); see \cite{TDB2022}
    \item West African Development Bank (BOAD); see \cite{BOAD2022}
\end{itemize}}

%To this end, we build on  the financial statement \cite{IBRD2022} of the International Bank for Reconstruction and Development (IBRD) as of June 2021 as well as on publicly available data from the Bank of Canada and Bank of England\footnote{Compare \url{https://www.bankofcanada.ca/2021/07/staff-analytical-note-2021-15/}.} which contains information on the International Bank for Reconstruction and Development's (IBRD) sovereign exposure data broken down by country.

We consider two different values for $\ELGD$ by either setting $\ELGD=45\%$ or, to account for preferred creditor treatment (PCT), by choosing $\ELGD=10\%$ which is in line with the specification imposed by S\&P for sovereign's borrowing from MDBs with high preferred creditor status. Moreover, we set the parameter $\nu=25\%$ implying a volatility of the LGD variable of $\VLGD=0.25$ without PCT (i.e. when $\ELGD=45\%$) and $\VLGD=0.15$ with PCT (i.e. for $\ELGD=10\%$). This choice is in line with estimates of sovereign haircuts in \cite{CrucesTrebesch2013} when assuming the recovery rate to be beta distributed. \\
%Note: For the stochastic LGDs we use the estimates in Cruces and Trebesch (2013) for mean and standard deviations of sovereign haircuts equal to 37.04\% and 27.28\%, resp.. \cite{RiskControl2023} assumes a beta distributed recovery rate with mean $\theta$ and volatility $\lambda\cdot\sqrt{\theta\cdot(1-\theta)}$ and calculates $\lambda=0.2728/\sqrt{0.3704\times (1-0.3704}=0.56$ from the haircut estimates. Moreover, $\theta$ is set to 90\% with PCT and to 55\% without PCT adjustment so that volatilities are 0.168 with PCT and 0.279 without PCT. Our choice is consistent with this approach.\\\\

Since most loans to sovereign borrowers of various countries are denominated in USD, we use the average one-year foreign currency rating transition matrix estimated for the time period 1975--2021 as published in \cite{S&P2022}.\footnote{See \url{https://www.spglobal.com/ratings/en/research/articles/210412-default-transition-and-recovery-2020-annual-sovereign-default-and-rating-transition-study-11888070} (Table 35).} We follow
\cite{S&P2017,S&P2018}, who impose a minimal rating of B-, and in line with \cite{RiskControl2023} we merge all ratings CCC+ and worse into a consolidated rating `Cs'. 
%, i.e. in we add up all rating transition probabilities for ratings CCC+ up to CC in the transition matrix of \cite{S&P2022}. This is also in line with \cite{S&P2017,S&P2018} where a minimal rating of B- is imposed. 
%Further, since \cite{S&P2022} also consider transitions to a non-rated (`NR') category which we omit, we normalize their transition matrix by dividing all entries by 1 minus the NR transition probability. 
The resulting normalised transition matrix is reported in Table B.1 in \cite{LuetkebohmertSesterShen2023}.
To convert these historical transition probabilities into risk-neutral transition rates, we follow \cite{Agrawal2004} and \cite{Kealhofer2003} and assume a market Sharpe ratio of $\psi=0.4.$
Risk-neutral probabilities $p^*_s(t,T)$ are then calculated as in the KMV model by
$$
p^*_s(t,T)=\Phi(\Phi^{-1}(p_s(t,T))+\psi \sqrt{T-t} \sqrt{\rho}),
$$
where $\rho$ denotes the asset correlation. 
We determine the borrower-specific asset correlation parameter $\rho_n$ according to the formula in the IRB approach (compare \cite{Basel2011}, p.39) such that $\rho_n$ ranges between 12\% and 24\%.\\

The DL approach in the actuarial setting further requires the factor loading $\omega_n$ for each borrower $n$ as input, which we obtain by setting the UL capital in the IRB approach (compare \cite{Basel2011}, p. 39) equal to the UL capital in the CreditRisk$^+$ model (see equation (\ref{equ UL capital CR+})) with $x_q$ equal to the $q^{\rm th}$ quantile of the Gamma$(\xi,1/\xi)$-distributed risk factor $X$ for $\xi=0.25$ and then solve for $\omega_n$.
%Given the variance $V_s$ of the historical default rate time series for a given rating grade $s$, the factor loading can be estimated as $\omega_s=\sqrt{\xi\cdot V_s}/\PD_s$ where $1/\xi$ is the variance of the Gamma-distributed risk factor $X$ (compare \cite{Gordy2000}, equation (18)). From the Bank of Canada and Bank of England database, we estimate $V=0.016$ and we set $\xi=0.25$. Since we do not have enough data to obtain stable estimates for rating grade specific average PDs and variances of  historical default rates, we can only get a very broad estimate of the factor loading by taking the average PD and variance over all borrowers. This leads to a factor loading of approximately 0.6 for the CreditRisk$^+$ model.\\
%(here $\PD$ needs to be chosen rather conservative and as a representative of aggregate borrowers\\
Note that the asymptotic GA in the actuarial approach as derived in \cite{GordyLuetkebohmert2013} and implemented under the S\&P approach does not require this parameter as an input because it directly uses the $\mathcal{K}_n$ and $\mathcal{R}_n$ values of the IRB approach.

Detailed information on loan maturities is not publicly available but a rough distribution of exposures in different maturity ranges for the individual MDBs can be inferred from their financial statements and average maturities are reported Table 1 in \cite{LuetkebohmertSesterShen2023}. In our calculations, for simplicity, we set the maturity equal to one year for all loans. 
Further, we impose the assumption that interest is paid semiannually in all contracts and assume that coupon rates are equal to 1\%.

As risk-free interest rates used for discounting future cash flows we construct the Nelson-Siegel-Svensson curve for the US treasury rate based on the parameters published on the website of the Board of Governors of the Federal Reserve System.\footnote{See \url{www.federalreserve.gov/data/yield-curve-tables/feds200628_1.html}} Since we consider financial statements as of December 31, 2022, we construct the yield curve for December 30, 2022 (last business day of the year). 
%We can also use the US LIBOR curve for discounting.\footnote{As noted in the financial statement of IDB, the bank will switch to a new benchmark rate starting in 2022. Therefore, we consider the LIBOR rate for the stylized portfolios constructed from 2021 financial statements.}\\

We calculate the GA at a confidence level of $q=99.9\%$ over a time horizon of $T=1$ year. 
We compute the NN GA and the analytic approximation GA in both the actuarial CreditRisk$^+$ model and the MtM approach and calculate the percentage error with respect to the exact GA obtained by MC simulations with IS.\\
%When calculating the analytic GA as applied under the S\&P approach, we set the precision parameter $\xi=0.25$ and the asset correlation as specified in the IRB approach which maps the borrower PD to an asset correlation between 12\% and 24\%. This choice agrees with \cite{S&P2017,S&P2018}. For comparison with the NN GA, we set the factor loading $\omega_n$ such that the UL requirement $\mathcal{K}_n$ in the CreditRisk$^+$ model agrees with the one under the IRB approach. Moreover, the S\&P approach assumes a maturity of one year for all loans. We use this specifications for comparison with our own calculations of name concentration risk. 
%\begin{remark}
%  We could also calibrate the parameter $\xi$ by matching the GA for a representative MDB portfolio to the name concentration risk calculated by MC simulations in a default-only mode approach. The calibrated $\xi$ can then be used in our actuarial GA calculations.  \\
%\end{remark}

The GAs for the described MDB portfolios are reported in Tables~\ref{tbl_mdb_gas_045} and~\ref{tbl_mdb_gas_01}. Our results for the percentage error between the NN GA and the exact MC GA show that the NN approach is highly accurate for both the actuarial CreditRisk$^+$ model and the MtM CreditMetrics approach. Comparing the results to the percentage error of the approximate GAs documents that our NN GA clearly outperforms the respective analytical approximations. For the CreditRisk$^+$ model, the mean absolute percentage error is 119 for the approximate GA and only 13 for the NN GA in the case of $\ELGD=45\%$. The difference between the percentage errors is even greater (322 compared to 25) in the case of $\ELGD=10\%$. For individual portfolios, the accuracy of the GA calculation improves by a factor of up to 75 for the NN GA compared to the approximate GA for the case $\ELGD=45\%$ and even more (up to a factor of 146) for $\ELGD=10\%$. While the approximate GA is generally more accurate in the MtM setting with an average percentage error of 21, the NN GA still has a much lower mean absolute percentage error of only 12 in the MtM case for $\ELGD=45\%$ (and 36 compared to 18 in the case of $\ELGD=10\%$). For individual portfolios, the NN GA improves accuracy by a factor of up to 7 compared to the approximate GA in the MtM setting.

The NN GA is also highly accurate for the two very small portfolios in the sample which have less than 10 obligors. This is particularly remarkable since the NN was only trained on portfolios with at least 10 borrowers.

The out-of-sample results for the simulated portfolios in Section \ref{sec deep learning} already illustrated the superior performance of our NN GA compared to the analytic methods. 
The out-of-sample results based on real portfolio data in this subsection document that our NN GA is also highly accurate in real data applications, i.e. when applied to portfolios with characteristics that are not just sampled from the training distributions. The accuracy of the NN GA of course increases when the distributions used for training fit more closely to the empirical distributions of the real portfolios. While our NN was trained using parametric distributions with parameters fit to the real portfolios, training based on empirical distributions of historical portfolio data could substantially increase precision. In this way, MDBs or rating agencies, which have access to such historical portfolio data, could use our methodology to obtain a highly accurate and very fast estimate for the name concentration risk in current portfolios.

\begin{table}[h!]
\resizebox{\textwidth}{!}{
\begin{tabular}{@{}clllllllllll@{}}
\toprule
\multicolumn{1}{l}{}                 & CAF    & ADB    & AFDB   & IDB    & CDB    & CABEI  & EADB   & IBRD   & TDB    & BOAD   & EBRD   \\ \midrule
\textbf{Borrowers}                   & 16  & 38  & 29  & 26  & 16  & 11  & 4   & 78  & 21  & 8   & 37  \\
\textbf{$GA^{\rm NN,act}$}                   & 15.00                            & 7.61                              & 8.44                             & 9.45                             & 8.38                               & 17.46                             & 19.51                             & 4.17                             & 13.16                             & 16.04                             & 7.69  \\
\textbf{$GA^{\rm 1st,act}$}               & 29.52                            & 19.84                             & 15.63                            & 25.08                            & 22.05                              & 59.95                             & 49.97                             & 6.85                             & 34.62                             & 32.93                             & 14.29 \\
\textbf{$GA^{\rm NN,MtM}$}                   & 21.13                            & 11.88                             & 13.07                            & 16.18                            & 18.02                              & 34.73                             & 33.66                             & 5.32                             & 26.71                             & 24.50                             & 11.34 \\
\textbf{$GA^{\rm 1st,MtM}$}               & 23.87                            & 13.87                             & 12.36                            & 17.92                            & 19.64                              & 44.54                             & 42.86                             & 5.19                             & 31.25                             & 27.79                             & 10.91 \\
\textbf{$GA^{\rm MC,act}$}                   & 12.46                            & 7.07                              & 8.00                             & 9.24                             & 11.89                              & 20.38                             & 27.13                             & 4.35                             & 15.23                             & 15.78                             & 8.76  \\
\textbf{$GA^{\rm MC,MtM}$}                   & 18.71                            & 10.34                             & 12.48                            & 13.32                            & 18.76                              & 25.85                             & 36.44                             & 6.38                             & 25.95                             & 22.75                             & 11.20 \\
\textbf{Percentage Error $GA^{\rm NN,act}$}     & 20.42                            & 7.59                              & 5.54                             & 2.26                             & 29.50                              & 14.33                             & 28.08                             & 4.23                             & 13.58                             & 1.64                              & 12.21 \\
\textbf{Percentage Error $GA^{\rm 1st,act}$} & 136.90                           & 180.64                            & 95.32                            & 171.46                           & 85.49                              & 194.18                            & 84.18                             & 57.48                            & 127.33                            & 108.66                            & 63.09 \\
\textbf{Percentage Error $GA^{\rm NN,MtM}$}     & 12.91                            & 14.93                             & 4.71                             & 21.49                            & 3.97                               & 34.37                             & 7.63                              & 16.68                            & 2.94                              & 7.71                              & 1.21  \\
\textbf{Percentage Error $GA^{\rm 1st,MtM}$} & 27.56                            & 34.13                             & 0.93                             & 34.52                            & 4.69                               & 72.32                             & 17.62                             & 18.58                            & 20.42                             & 22.16                             & 2.62\\
\textbf{Acc. improv. factor (act)} & 6.70 & 23.80 &17.21 & 75.87 & 2.90 & 13.55 & 3.00 & 13.59& 9.38 & 66.26 &5.17\\  
\textbf{Acc. improv. factor (MtM)} & 2.13 & 2.29 & 0.20 & 1.61 & 1.18 & 2.10 & 2.31 & 1.11 & 6.95 & 2.87 & 2.17
\end{tabular}}
\caption{The table reports the GAs (in \% of total exposure) for the MDB portfolios described in Section~\ref{sec real data}, using the different approaches under the assumption that $\operatorname{ELGD} \equiv 0.45$ for all obligors. The percentage error is calculated as $|GA^{\rm NN, act}-GA^{\rm MC, act}|/GA^{\rm MC, act}$, and correspondingly for the other GAs. The accuracy improvement factor is calculated as $|GA^{\rm 1st, act}-GA^{\rm MC, act}|/|GA^{\rm NN, act}-GA^{\rm MC, act}|$, resp. for the MtM GAs.}\label{tbl_mdb_gas_045}
\end{table}

\begin{table}[h!]
\resizebox{\textwidth}{!}{
\begin{tabular}{@{}clllllllllll@{}}
\toprule
\multicolumn{1}{l}{}                 & CAF    & ADB    & AFDB   & IDB    & CDB    & CABEI  & EADB   & IBRD   & TDB    & BOAD   & EBRD   \\ \midrule
\textbf{Borrowers}                   & 16  & 38  & 29  & 26  & 16  & 11  & 4   & 78  & 21  & 8   & 37  \\
\textbf{$GA^{\rm NN,act}$}                   & 4.46                             & 3.67                              & 3.88                             & 4.15                             & 3.77                               & 6.58                              & 5.62                              & 2.33                             & 5.61                              & 5.14                              & 3.03   \\
\textbf{$GA^{\rm 1st,act}$}               & 20.35                            & 13.93                             & 10.68                            & 17.51                            & 14.56                              & 41.84                             & 29.58                             & 4.57                             & 24.81                             & 22.77                             & 9.61   \\
\textbf{$GA^{\rm NN,MtM}$}                   & 15.46                            & 8.86                              & 8.38                             & 12.19                            & 12.36                              & 25.64                             & 20.49                             & 3.66                             & 22.19                             & 18.69                             & 6.94   \\
\textbf{$GA^{\rm 1st,MtM}$}               & 17.93                            & 10.63                             & 9.03                             & 13.61                            & 13.62                              & 34.11                             & 25.42                             & 3.60                             & 26.85                             & 21.42                             & 7.59   \\
\textbf{$GA^{\rm MC,act}$}                   & 4.35                             & 2.65                              & 3.13                             & 3.38                             & 3.96                               & 5.98                              & 12.72                             & 1.49                             & 6.88                              & 5.58                              & 2.33   \\
\textbf{$GA^{\rm MC,MTM}$}                   & 11.58                            & 7.78                              & 8.09                             & 10.19                            & 11.72                              & 18.88                             & 22.12                             & 3.74                             & 15.98                             & 13.39                             & 6.84   \\
\textbf{Percentage Error $GA^{\rm NN,act}$}     & 2.52                             & 38.66                             & 24.03                            & 22.72                            & 4.81                               & 10.02                             & 55.79                             & 56.28                            & 18.45                             & 7.95                              & 29.96  \\
\textbf{Percentage Error $GA^{\rm 1st,act}$} & 367.73                           & 425.79                            & 241.11                           & 417.94                           & 267.74                             & 599.68                            & 132.55                            & 206.47                           & 260.61                            & 308.03                            & 312.35 \\
\textbf{Percentage Error $GA^{\rm NN,MtM}$}     & 33.51                            & 13.92                             & 3.60                             & 19.63                            & 5.48                               & 35.78                             & 7.36                              & 2.22                             & 38.87                             & 39.58                             & 1.46   \\
\textbf{Percentage Error $GA^{\rm 1st,MtM}$} & 54.80                            & 36.69                             & 11.56                            & 33.52                            & 16.22                              & 80.65                             & 14.91                             & 3.61                             & 68.05                             & 59.97                             & 11.01 \\
\textbf{Acc. improv. factor (act)} & 145.92 &11.01 & 10.03 & 18.40 & 55.66 & 59.85 & 2.38 & 3.67 & 14.13 &38.75 & 10.43\\  
\textbf{Acc. improv. factor (MtM)} & 1.64 & 2.64 & 3.21 & 1.71 & 2.96& 2.25 & 2.03 & 1.63 & 1.75 & 1.56& 7.54
\end{tabular}}
\caption{The GAs of the MDBs, described in Section~\ref{sec real data}, using the different approaches under the assumption that $\operatorname{ELGD} \equiv 0.1$ for all obligors. The percentage error is calculated as $|GA^{\rm NN, act}-GA^{\rm MC, act}|/GA^{\rm MC, act}$, and correspondingly for the other GAs. The accuracy improvement factor is calculated as $|GA^{\rm 1st, act}-GA^{\rm MC, act}|/|GA^{\rm NN, act}-GA^{\rm MC, act}|$, resp. for the MtM GAs.}\label{tbl_mdb_gas_01}
\end{table}

\subsection{Application to Stress  Testing}

The DL approach presented in this paper offers a readily applicable framework to study the effects of changes in the underlying portfolio characteristics on the GA, compare also Section~\ref{sec sensitivity analysis}. A key application of this feature is stress testing -- specifically, examining how name concentration risk in loan portfolios responds to stressed market environments. 

Stress testing comes in different levels. At the basic level, scenarios are analysed in which individual parameters are varied while all others remain unchanged. This provides an initial indication of the sensitivity to potential risk factors. More advanced stress testing techniques stress multiple parameters simultaneously, taking into account potential interactions between risk factors leading to tail events. \cite{Basel2009stresstesting,Basel2018stresstesting} points to the importance of studying the impact of such recession-type scenarios with system-wide interactions and feedback effects, where scenarios can be based on historical data or hypothetical. Finally, reverse stress tests do the exact opposite, analysing which events lead to adverse outcomes.

For the construction of stressed scenarios for the MDB portfolios, we follow \cite{HardySchmieder2023}, who show how the non performing loans (NPL) stock ratio as an indicator of credit risk increases during crisis periods for countries from different regions. The authors find that loss rates (i.e. the product of PD and LGD) increase by 1.5 percentage points (\%pt) in the medium stress scenarios and by 5.5\%pt up to 12.1\%pt in the severe-to-extreme stress scenarios for emerging economies from a normal average level of around 1.9\%. In addition, the authors report that LGDs also vary with the business cycle and that LGDs and default probabilities are positively correlated, with LGDs being less volatile than PDs. For advanced economies they document an average LGD of 26\% in normal periods, 34\% in medium and 41\% in severe stress scenarios. Due to a lack of data, they do not provide LGD levels for emerging economies for different market scenarios but only historical average LGD rates, which range from 59\% to 62\%. Nevertheless, the authors suggest that the absolute increase in LGDs as estimated for advanced economies (around 15\%pt in severe stress scenarios) should be added to the historical average LGDs for emerging markets to obtain stressed scenarios, resulting in stressed LGDs (excluding PCT effects) of around 75\%.
Default rates can then be approximated by dividing the loss rates by the LGD rates for each stress scenario, implying average PDs of around 3\% under normal conditions and an increase of 7\%pt in severe stress scenarios and 18\%pt in extreme scenarios. Thus, our stress test example considers similar increases in PDs and, for severe scenarios, a joint increase in LGDs.
Typically default rates and LGD rates peak in times of economic downturn (compare \cite{Moodys2016}). Thus, we can interpret the scenarios with high default and LGD rates also as scenarios with decreasing GDP growth rates and increasing unemployment rates. 

%In the IRB model, asset correlations are determined by PDs so that under stress scenarios with increased PDs, asset correlations decrease (which usually implies an increase in name concentration risk). As pointed out in \cite{HardySchmieder2023}, this relationship can be misleading. Instead, in line with \cite{DuellmannScheicherSchmieder} empirical asset correlations tend to increase during economic downturns with asset correlations up to 30\% under stress scenarios compared to around 10\% under normal conditions. For this reason, we also include a sensitivity analysis studying the impact of increasing asset correlations on the GA.  \textcolor{red}{Here we should include a figure showing in panel a) the impact of increasing PD and correspondingly decreasing IRB asset correlations on the GA in the MtM setting and in panel b) the effect when PDs are increased and at the same time asset correlations are increasing at a similar rate as the decrease in panel a. }\\

We turn first to the basic stress test. Here we examine the effect on the GA of a gradual increase in the PD for all obligors at the same time. In the case of the rating-based MtM approach, this is equivalent to a simultaneous downgrade of all obligors in the loan portfolio. The asset correlation $\rho_n$ is calculated by applying the IRB formula to the corresponding increased PDs. In Figure~\ref{fig:stress_pds} we consider the resulting impact on the GA for three MDB portfolios (AFDB, CAF and EBRD) 
%as well as on the synthetic portfolios (referred to as base portfolio), whose compositions are illustrated in Figure~\ref{fig_sensitivities_actuarial}~(b) and Figure~\ref{fig_sensitivities_mtm}~(b), 
for the actuarial and MtM cases, respectively. A simultaneous increase in the default risk of all borrowers in most cases leads to a slight increase in the GA in the actuarial case. Overall default risk increases but the portfolio does not necessarily become significantly more concentrated. In the MtM case, the effect is more pronounced as downgrades affect not only the repayment of the nominal amount at maturity but also MtM losses during the holding period. In addition, a downgrade can result in a much larger increase in PD for some borrowers than for others, e.g. a downgrade from BBB+ to BBB implies an increase in PD of 0.02\%pt, whereas a downgrade from BB+ to BB corresponds to an increase in PD of 22.\%pt. As a result, portfolios can become much more concentrated from simultaneous downgrades than from simultaneous increases in PDs.
%As already observed in Figures \ref{fig_sensitivities_actuarial} and \ref{fig_sensitivities_mtm}, deviations in the creditworthiness of individual borrowers can impact the GA in both directions. Therefore, we also observe situations where a modest simultaneous downgrade of all borrowers can result in a slight reduction of the name concentration risk in the portfolio. This is visible in Figure \ref{fig:stress_pds} for the ADB and CABEI portfolios.\\

\begin{figure}
    \centering
    \includegraphics[width=0.95\linewidth]{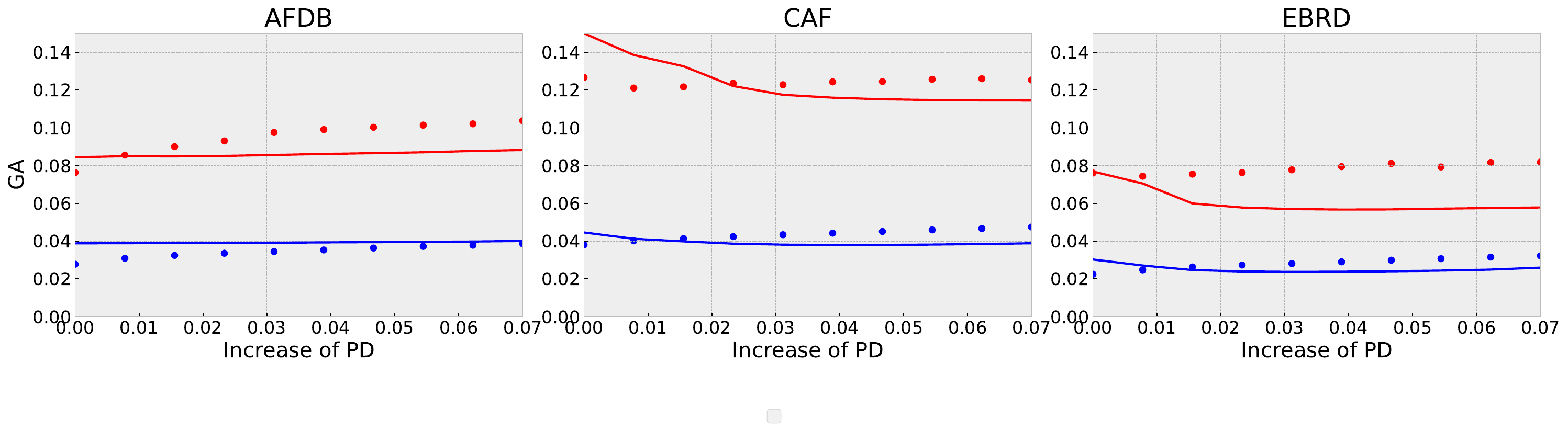}
        \includegraphics[width=0.95\linewidth]{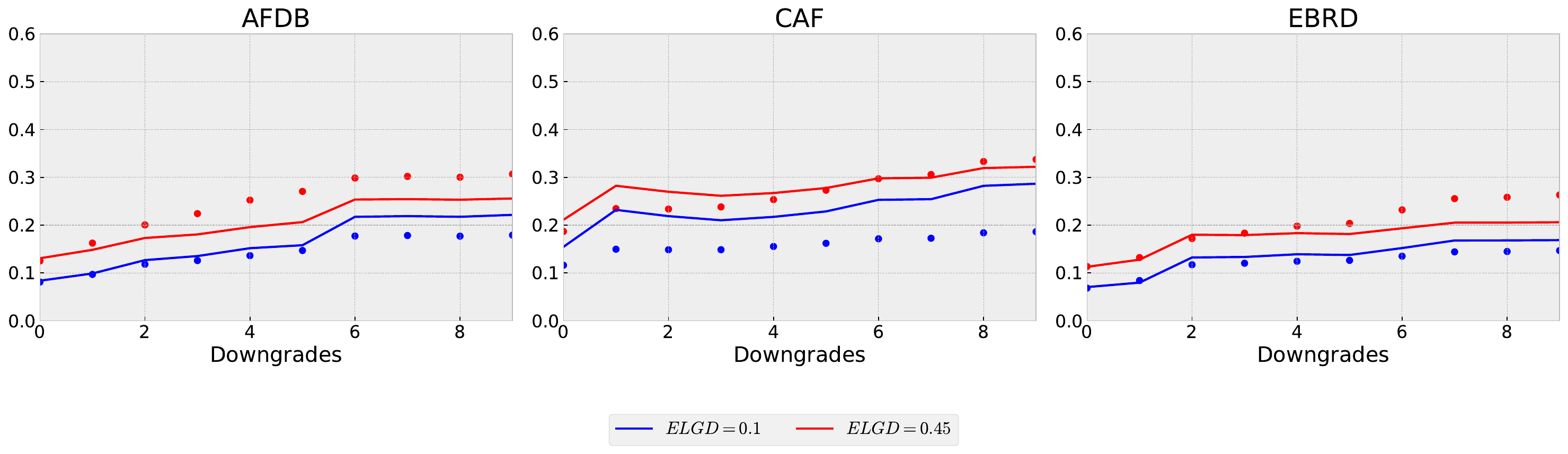}
    \caption{ The figure shows how the GA changes in response to a gradual increase in the PD (in \%pt) of all obligors at the same time. Top: actuarial approach, Bottom: MtM approach. Solid curves: NN GAs, dotted curves: MC GAs.}
    \label{fig:stress_pds}
\end{figure}

Next, we examine the effect of an increase in exposure concentration on the GA. To realize an increase in exposure concentration we readjust the portfolio exposure shares according to 
\begin{equation}\label{eq:ead_scaling}    
a_n' =  \frac{a_n^\alpha}{\sum_{i=1}^N a_i^\alpha} \text{ for } n =1,\dots,N,
\end{equation}
where the exponent $\alpha$ is determined such that 
\[
\sum_{n=1}^N\left(\frac{a_n'}{\sum_{i=1^N}{a_i'}}\right)^2 = h
\]
for a given level $h$. This means we rescale the exposures such that the Herfindahl-Hirschman index (HHI), see \cite{rhoades1993herfindahl}, of the rescaled portfolio attains a certain level $h$.
Figure~\ref{fig:stress_hhi} shows the results for the four MDB portfolios, when increasing the HHI to a level $h$ shown on the $x$-axis. As expected the GAs generally increase with increasing exposure concentration. The NN approach captures this effect and is generally quite accurate showing the robustness of our approach to varying market environments. In some cases (e.g. the MtM GA for the CAF portfolio), the NN GA overestimates the true MC GA. This phenomenon can be attributed to the fact that the NN was trained on parametric distributions, specifically those with exponentially distributed exposures. These training distributions may not accurately represent the observed exposure distribution of the portfolio. To further analyse this effect, we computed the Wasserstein distances between the observed exposure distributions and the maximum likelihood-estimated exponential distributions, based on on 10,000 samples from the latter. The findings reveal that the deviation between the NN GA and the MC GA generally increases with larger Wasserstein distances. For example, the Wasserstein distance between exposure distribution and estimated exponential distribution for the CAF portfolio is 0.0198, compared to 0.0090 for AFDB and 0.0060 for EBRD. However, the robustness of the NN GA can be enhanced by expanding the training dataset to include also stressed market scenarios, as considered in this analysis.\\

\begin{figure}
    \centering
    \includegraphics[width=0.95\linewidth]{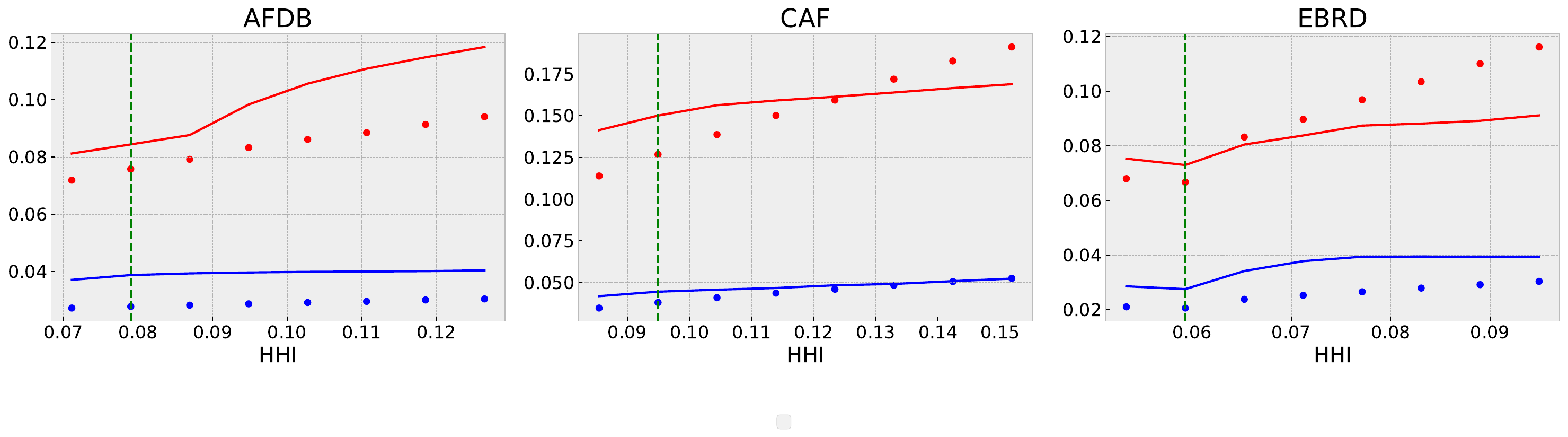}
    \includegraphics[width=0.95\linewidth]{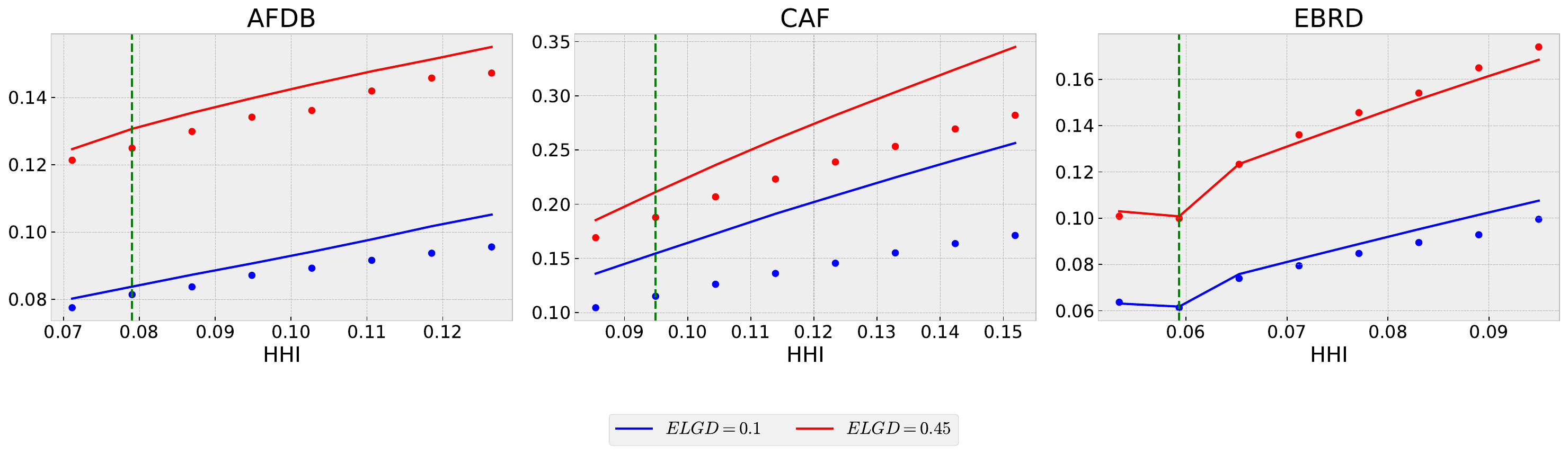}
    \caption{ The figure shows the effect on the GA of simultaneously increasing the exposure concentration of the three MDB portfolios under consideration by setting the HHI to a fixed level $h$ shown on the $x$-axis. Top: actuarial approach, Bottom: MtM approach. Green dashed line: observed HHI level. Solid curves: NN GAs, dotted curves: MC GAs.}
    \label{fig:stress_hhi}
\end{figure}

Overall, we observe that both increasing probabilities of default and increasing exposure concentrations, considered in isolation, lead to an increase in the resulting GA in most cases. However, in realistic stress scenarios, PDs and exposure concentrations may increase simultaneously. We therefore describe four different stress scenarios below that take this into account. The stress parameters are chosen in accordance with the literature discussed above. 

\begin{itemize}
    \item \textbf{Scenario 1, observed case}:\\ PD as observed, $\operatorname{EAD}$ as observed, $\ELGD = 0.1$;
    \item \textbf{Scenario 2, medium stress}:\\ PD increased by $2 \%$ in the actuarial case, downgrade of 1 notch in the MtM case. $\operatorname{EAD}$s are scaled according to \eqref{eq:ead_scaling} such that $\operatorname{HHI}$ is increased by $25\%$, $\ELGD = 0.1$;
    \item \textbf{Scenario 3, severe stress}:\\ PD increased by $4 \%$ in the actuarial case, downgrade of 2 notch in the MtM case. $\operatorname{EAD}$s are scaled according to \eqref{eq:ead_scaling} such that $\operatorname{HHI}$ is increased by $50\%$, $\ELGD = 0.45$;
    \item \textbf{Scenario 4, extreme stress}:\\ PD increased by $7 \%$ in the actuarial case, downgrade of 3 notches in the MtM case. $\operatorname{EAD}$s are scaled according to \eqref{eq:ead_scaling} such that  $\operatorname{HHI}$ is increased by $100\%$, $\ELGD = 0.45$;
\end{itemize}
Asset correlations are calculated using the IRB formula with the PDs in the corresponding scenarios. 
Hence, the scenarios reflect different stress levels and account for various interactions and feedback effects between individual variables. Table \ref{tbl:stress_mtm} shows the NN GA for the eleven MDB portfolios in the different stress scenarios. As expected, in almost all cases, the GA increases as the stress increases. The modest decrease in the GA in Scenario 2 compared to Scenario 1 for the EBRD and BOAD portfolios is due to an approximation error between the NN GA and the MC GA. Overall, the mean absolute percentage error in the medium stress scenario is comparable to that in the observed case, though it increases slightly in the severe and extreme stress scenarios. This increase can be attributed to the growing divergence between the exposure and PD distributions used to train the NN and those specified in the stressed market environments. The performance can be improved by extending the training data set. This is particularly relevant in the context of reverse stress testing, as our results help to identify scenarios that lead to heightened name concentration. These scenarios can then be incorporated into the training dataset to enhance the accuracy of computations under such stressed market conditions.

% Actuarial Case Table
\begin{table}[htb!]
\centering
\resizebox{\textwidth}{!}{
\begin{tabular}{@{}clllllllllll@{}}
\toprule
\multicolumn{1}{l}{\textbf{Actuarial}} & \multicolumn{1}{r}{\textbf{CAF}} & \multicolumn{1}{r}{\textbf{ADB}} & \multicolumn{1}{r}{\textbf{AFDB}} & \multicolumn{1}{r}{\textbf{IDB}} & \multicolumn{1}{r}{\textbf{CDB}} & \multicolumn{1}{r}{\textbf{CABEI}} & \multicolumn{1}{r}{\textbf{EADB}} & \multicolumn{1}{r}{\textbf{IBRD}} & \multicolumn{1}{r}{\textbf{TDB}} & \multicolumn{1}{r}{\textbf{BOAD}} & \multicolumn{1}{r}{\textbf{EBRD}} \\ \midrule
\textbf{Scenario 1}           & 4.46                             & 3.67                             & 3.88                              & 4.15                             & 3.77                             & 6.58                               & 5.62                              & 2.33                              & 5.61                             & 5.14                              & 3.03                              \\
\textbf{Scenario 2}           & 4.57                             & 3.78                             & 4.03                              & 4.15                             & 3.91                             & 6.73                               & 16.72                             & 2.38                              & 6.91                             & 4.71                              & 2.97                              \\
\textbf{Scenario 3}           & 15.70                            & 8.10                             & 10.63                             & 10.50                            & 13.42                            & 18.69                              & 24.66                             & 5.07                              & 17.13                            & 16.45                             & 9.10                              \\
\textbf{Scenario 4}           & 16.73                            & 14.55                            & 13.30                             & 14.55                            & 15.36                            & 21.47                              & 28.83                             & 5.59                              & 19.63                            & 15.73                             & 10.17                             \\ \hline
\bottomrule
\midrule
\multicolumn{1}{l}{\textbf{MtM}} & \multicolumn{1}{r}{\textbf{CAF}} & \multicolumn{1}{r}{\textbf{ADB}} & \multicolumn{1}{r}{\textbf{AFDB}} & \multicolumn{1}{r}{\textbf{IDB}} & \multicolumn{1}{r}{\textbf{CDB}} & \multicolumn{1}{r}{\textbf{CABEI}} & \multicolumn{1}{r}{\textbf{EADB}} & \multicolumn{1}{r}{\textbf{IBRD}} & \multicolumn{1}{r}{\textbf{TDB}} & \multicolumn{1}{r}{\textbf{BOAD}} & \multicolumn{1}{r}{\textbf{EBRD}} \\ \midrule
\textbf{Scenario 1}           & 15.46                            & 8.86                             & 8.38                              & 12.19                            & 12.36                            & 25.64                              & 20.49                             & 3.66                              & 22.19                            & 18.69                             & 6.94                              \\
\textbf{Scenario 2}           & 31.02                            & 10.15                            & 10.79                             & 16.10                            & 30.98                            & 34.11                              & 35.60                             & 4.31                              & 28.07                            & 17.94                             & 9.89                              \\
\textbf{Scenario 3}           & 43.99                            & 15.50                            & 20.45                             & 24.28                            & 43.23                            & 67.75                              & 149.65                            & 9.05                              & 43.65                            & 32.04                             & 27.81                             \\
\textbf{Scenario 4}           & 57.62                            & 18.40                            & 23.91                             & 30.73                            & 57.62                            & 94.17                              & 196.83                            & 10.65                             & 57.23                            & 54.38                             & 35.97                             \\ \bottomrule
\end{tabular}
}
\caption{NN GA (in \%) for the MDB portfolios under the different stress scenarios for both the actuarial case (upper panel) and the MtM case (lower panel).}\label{tbl:stress_mtm}
\end{table}

%Finally, with regard to reverse stress testing, our previous results suggest that higher exposure concentration and rising LGDs are the main drivers of increased GAs in stressed market environments. In addition, lower asset correlation intensifies concentration risk.

\section{Conclusion}\label{sec conclusion}

%The risk adjusted capital (RAC) framework implemented by S\&P to assess MDBs' capital adequacy is calculated by dividing shareholder equity by risk-weighted assets (RWA). An adjustment for single name concentration risk is added to RWA so that RAC decreases when concentration risk increases. Thus, MDBs need to increase their equity to maintain the same RAC. The impact of concentration risk adjustment on MDBs RAC is illustrated in Fig. 1 and 3 in \cite{Humphrey2015} and can be substantial (potentially doubling the RWA) for regional MDBs that typically lend to a small number of sovereign borrowers.  This in turn limits MDBs lending headroom and restrict MDBs in achieving their development goals.

In this paper, we proposed a new DL-based approach for the quantification of name concentration risk in small portfolios. While analytic approximations as, e.g., the GA suggested in \cite{GordyLuetkebohmert2013} and currently applied by S\&P have the strong advantage of being easily implementable, they were originally developed for commercial bank portfolios which are much more diversified than those of specialized institutions such as MDBs. As a consequence, these approaches may significantly overestimate name concentration risk when applied to portfolios with very few borrowers. In contrast, our NN-based GA is specifically designed for these very small portfolios, enabling a highly accurate and rapid measurement of concentration risk. 
In particular, our approach allows for precise evaluation of how variations in individual borrower characteristics impact the portfolio’s name concentration risk, as well as how this risk shifts when additional loans are incorporated into an existing portfolio. Furthermore, our methodology offers an efficient framework for conducting stress testing applications.

\bibliographystyle{agsm}
	\bibliography{Literature}

	\appendix

\newpage 

\part*{\protect\large E-companion to ``Measuring Name Concentrations through Deep Learning''}

\setcounter{page}{1}

%\numberwithin{equation}{section}
\renewcommand*{\thesection}{EC.\arabic{section}} 

\section{Review of GA methodology}\label{app GA methodology}

An analytical expression for the GA can be obtained as a first order asymptotic approximation of the difference (\ref{equ def CR}) and can be expressed as (compare e.g. \cite{GordyLuetkebohmert2013} and \cite{Voropaev})
$$
\begin{array}{ccl}
\alpha_q(L)-\alpha_q(\mathbb{E}[L|X])
&\approx&-\frac{1}{2} \frac{1}{h(x)}\frac{d}{dx}\left[\frac{h(x)}{m_1'(x)}m_2(x)\right]\Big|_{x=x_q} 
\end{array}
$$
where $x_q=\alpha_q(X)$ is the $q^{\rm th}$ quantile and $h(\cdot)$ is the probability density function (pdf) of the systematic risk factor $X$, and where $m_i(x)$ denotes the $i^{\rm th}$ conditional moment of the loss ratio $L$ given the systematic factor $X=x$. %The above formula can be calculated further as
%\begin{equation}\label{equ GA 1st order}
%\begin{array}{ccl}
%GA^{1st}&=&- \frac{1}{2m_1'(x_q)} \left[m_2'(x_q)+m_2(x_q)\cdot \left( 
%\frac{h'(x_q)}{h(x_q)}-\frac{m_1''(x_q)}{m_1'(x_q)}\right)\right].
%\end{array}
%\end{equation}

\subsection{GA in the actuarial approach}\label{app actuarial GA}

Assuming the CreditRisk$^+$ setting with Gamma-distributed risk factor $X$ (compare Section \ref{sec actuarual}), 
\cite{GordyLuetkebohmert2013} parameterize the inputs to the GA in terms of EL reserve requirement $$\mathcal{R}_n=\ELGD_n \PD_n$$ and UL capital requirement \begin{equation}\label{equ UL capital CR+}
\mathcal{K}_n=\ELGD_n\PD_n \omega_n(x_q-1)
\end{equation}
and then substitute these by the corresponding expressions in the Internal Ratings Based (IRB) approach in \cite{Basel2011}, so that the factor loadings $\omega_n$ do not need to be specified. They then derive the following analytic solution for the GA in the actuarial framework of CreditRisk$^+$ as
\begin{equation}\label{equ GA 1st order actuarial}
\begin{array}{ccl}
GA^{\rm 1st,act}&=&\displaystyle  \frac{1}{2\mathcal{K}^*} \sum_{n=1}^N a_n^2 \Big[\delta\left(\mathcal{C}_n (\mathcal{K}_n+\mathcal{R}_n)+(\mathcal{K}_n+\mathcal{R}_n)^2 \frac{\VLGD^2_n}{\ELGD_n^2}\right)\\
&&\displaystyle\hspace{1cm}  - \mathcal{K}_n \left(\mathcal{C}_n +2(\mathcal{K}_n+\mathcal{R}_n) \frac{\VLGD^2_n}{\ELGD_n^2}\right)\Big]
\end{array}
\end{equation}
where $\mathcal{C}_n=\frac{\VLGD_n^2+\ELGD_n^2}{\ELGD_n}$ and $\VLGD_n^2=\nu\cdot \ELGD_n(1-\ELGD_n)$ for some parameter $\nu\in(0,1).$ 
Moreover, $\mathcal{K}^*=\sum_{n=1}^N a_n \mathcal{K}_n$ and 
$$
\delta\equiv-(x_q-1) \frac{h'(x_q)}{h(x_q)}=(x_q-1)\left(\xi+\frac{1-\xi}{x_q}\right).
$$
By ignoring terms that are quadratic or of higher order in $\PD$s, the authors also derive a simplified version of the above GA which is implemented in the current S\&P methodology for capital adequacy of MDBs with precision parameter fixed to $\xi=0.25$.

\subsection{GA in the MtM approach}\label{app MtM GA}

\cite{GordyMarrone2012} show that the GA can be derived analytically for a large class of single-factor MtM  models and explicitly apply their methodology to CreditMetrics and KMV Portfolio Manager. In the CreditMetrics MtM model the GA can be formulated as 
{\begin{footnotesize}
\begin{equation}\label{eq_GA_formula_1st_order_mtm}
\begin{aligned}
GA^{\rm 1st, MtM}&=\frac{1}{2} e^{-rT} \left(-\alpha_q(X)\frac{
\sum_{n=1}^N a_n^2 \sigma_n^2(\alpha_q(X))
}{\sum_{n=1}^N a_n \mu_n'(\alpha_q(X))} 
+\frac{d}{dx} \left(\frac{\sum_{n=1}^N a_n^2 \sigma_n^2(x)}{
\sum_{n=1}^N a_n \mu_n'(x)
}\right)\Big|_{x=\alpha_q(X)}\right)\\
&=\frac{1}{2} e^{-rT} \Big(-\alpha_q(X)\frac{
\sum_{n=1}^N a_n^2 \sigma_n^2(\alpha_q(X))
}{\sum_{n=1}^N a_n \mu_n'(\alpha_q(X))} 
+\frac{\sum_{n=1}^N a_n^2 \frac{d}{dx} \sigma_n^2(x)|_{x=\alpha_q(X)}}{
\sum_{n=1}^N a_n \mu_n'(\alpha_q(X))
}\\
& \hspace{2cm} -  \frac{\left(\sum_{n=1}^N a_n^2 \sigma_n^2(\alpha_q(X))\right)\cdot\left(\sum_{n=1}^N a_n \mu_n''(\alpha_q(X))\right)}{
(\sum_{n=1}^N a_n \mu_n'(\alpha_q(X)))^2}\Big),
\end{aligned}
\end{equation}
\end{footnotesize}}
where the conditional expectation $\mu_n(x)$ is given by (\ref{equ expected return given x}) and the conditional variance is
$$
\sigma_n^2(x)\equiv \mathbb{V}[R_n|X=x] =\sum_{s=0}^S \left(\xi_{ns}^2(x)+\lambda_{ns}^2(x)\right)\cdot \pi_{ns}(x)-\mu_n^2(x),
$$
where $\xi_{ns}^2(x)=\mathbb{V}[R_n|X=x,S_n=s]$. % and $\Xi_n(x)$ is the vector of $\{\xi_{ns}(x)\}$ and $\Lambda_n^2(x)$ is the vector of $\{\lambda_{ns}^2(x)\}.$ 
For the calculation of the first and second derivative of the conditional expectation $\mu_n'(x)$ and $\mu_n''(x)$ we further need the derivatives of $\pi_{ns}(x)$ which can be instantly calculated as
\begin{align*} 
\pi_{ns}'(x)&=-\sqrt{\frac{\rho_n}{1-\rho_n}} \phi\left(\frac{C_{g(n),s}-x\sqrt{\rho_n}}{\sqrt{1-\rho_n}}\right)+\sqrt{\frac{\rho_n}{1-\rho_n}}\phi \left(\frac{C_{g(n),s-1}-x\sqrt{\rho_n}}{\sqrt{1-\rho_n}}\right)\\
\pi_{ns}''(x)&=-\frac{\rho_n}{1-\rho_n}\cdot \frac{C_{g(n),s}-x\sqrt{\rho_n}}{\sqrt{1-\rho_n}} \phi\left(\frac{C_{g(n),s}-x\sqrt{\rho_n}}{\sqrt{1-\rho_n}}\right)\\
&\hspace{2cm} +\frac{\rho_n}{1-\rho_n}\cdot \frac{C_{g(n),s-1}-x\sqrt{\rho_n}}{\sqrt{1-\rho_n}} \phi \left(\frac{C_{g(n),s-1}-x\sqrt{\rho_n}}{\sqrt{1-\rho_n}}\right)\\
\pi_{ns}'''(x) &=\left(\frac{\rho_n}{1-\rho_n}\right)^{3/2}\cdot \left(1-\left(\frac{C_{g(n),s}-x\sqrt{\rho_n}}{\sqrt{1-\rho_n}}\right)^2\right) \phi\left(\frac{C_{g(n),s}-x\sqrt{\rho_n}}{\sqrt{1-\rho_n}}\right)\\
	& +\left(\frac{\rho_n}{1-\rho_n}\right)^{3/2}\cdot \left(\left(\frac{C_{g(n),s-1}-x\sqrt{\rho_n}}{\sqrt{1-\rho_n}}\right)^2-1\right) \phi \left(\frac{C_{g(n),s-1}-x\sqrt{\rho_n}}{\sqrt{1-\rho_n}}\right).
\end{align*}
where we used the property $\phi'(x)=-x\phi(x).$

In CreditMetrics, it is assumed that the conditional variance $\xi_{ns}^2(x)=0$ for all $s\geq 1.$ Further, in the default state only idiosyncratic recovery risk is assumed so that $$\xi_{n0}^2(x)=\xi_{n0}^2=\mathbb{V}[R_n|S_n=0]=\left(\frac{F(T)}{P_{n0}}\right)^2 \VLGD^2.
$$

%\newpage
\section{Auxiliary Mathematical Results}

\begin{lemma}\label{lem_cont_distrib}
Let $L(Y)$ be a real-valued random variable on $(\Omega,\mathcal{F},\mathbb{P})$ for all $Y \in \mathbb{K}\subset \R^m$, $m\in \N$, with $\mathbb{K}$ compact. Let 
\begin{equation}\label{eq_map_lem}
\begin{aligned}
\mathbb{K} \times [0,\bar{\ell}]  \ni (Y,\ell) &\mapsto \mathbb{P}(L(Y) \leq \ell) \in \R
\end{aligned}
\end{equation}
be continuous in $(Y,\ell)$ and strictly increasing in $\ell$ on some interval $[0,\bar{\ell}]$. Then, for all $q\in (0,1)$ we have that 
\[
Y \mapsto \inf \left\{\ell\in[0,\bar{\ell}]:\mathbb{P}\left(L(Y) \leq \ell\right) \geq q \right\}=\alpha_q(L(Y))
\]
is continuous.
\end{lemma}

\begin{proof}
    Since $(\ell,Y)\mapsto F_{L(Y)}(\ell):=\mathbb{P}(L(Y) \leq \ell)$ is continuous in $(\ell,Y)$ and strictly increasing in $\ell$, we have that 
    $$\alpha_q(L(Y))=\inf\{\ell\in[0,\bar{\ell}]: \mathbb{P}(L(Y)\leq \ell)\geq q\}=F_{L(Y)}^{-1}(q)$$ is the inverse distribution function which is continuous in $q$. 
    Now let $(Y_i)_{i \in \N} \subseteq \mathbb{K}$ be a sequence with $Y_i\rightarrow Y$ as $i \rightarrow \infty$, and consider the sequence $(F_{L(Y_i)}^{-1}(q))_{i\in\mathbb{N}}$ for some fixed $q$. Since $\alpha_q(L(Y))$ is bounded for $Y\in\mathbb{K}$, there exists a subsequence, denoted by $\ell_{i_k}:=F_{L(Y_{i_k})}^{-1}(q)$ for $k\in\mathbb{N}$, which converges to some $\ell^*$. Then, we have that
    $$
    \begin{array}{ccl}
    F_{L(Y)}(\ell^*)&=&\lim_{k\to \infty} F_{L(Y_{i_k})}(\ell_{i_k}) \quad\mbox{(due to the continuity of F)}\\
    &=& \lim_{k\to \infty} F_{L(Y_{i_k})}(F_{L(Y_{i_k})}^{-1}(q))=q =F_{L(Y)}\left(F_{L(Y)}^{-1}(q)\right).
    \end{array}
    $$
    Since $F$ is strictly increasing, we thus obtain $\ell^*=F_{L(Y)}^{-1}(q)$ and the sequence $\{F_{L(Y_{i_k})}^{-1}(q)\}_{k\in\mathbb{N}}$ converges to $\ell^*=F_{L(Y)}^{-1}(q)$ which  proves the continuity of $Y\mapsto \alpha_q(L(Y)).$
\end{proof}

%\newpage

\section{Algorithms}\label{app algorithms}

\begin{algorithm}[h!]
\SetAlgoLined
%\KwData{Data;}
\SetKwInOut{Input}{Input}
\SetKwInOut{Output}{Output}

{\begin{footnotesize}
\Input{Number $K$ of simulations}

\begin{enumerate}
    \item Determine $\tau_y$ as in equation (\ref{equ tau_y});
    \item Determine $t$ as in equation (\ref{equ choosing t});
 \end{enumerate}
    \For{$k=1,\dots, K$}{
  \begin{enumerate}
    \setcounter{enumi}{2}
    \item Generate $X^{(k)}\sim \operatorname{Gamma}(\xi,\frac{1}{\xi -t})$; %in our case $k=\xi$ and $$\theta=1/\xi$
    \item Generate $\LGD_n^{(k)}$ according to a Beta-distribution with mean $\operatorname{ELGD}_n$ and variance $\nu\cdot \operatorname{ELGD}_n \cdot (1-\operatorname{ELGD}_n)$ for $n=1,\dots,N$;
    \item Generate default indicators from the tilted distribution $Y_n^{(k)}\sim \operatorname{Bernoulli} \left( \frac{\pi_n(X^{(k)}) e^{a_n \LGD_n \tau_y}}{1-\pi_n(X^{(k)})+\pi_n(X^{(k)})e^{a_n \LGD_n\tau_y }}\right)$ for $n=1,\dots,N$;%with $\pi_n(X^{(k)})=\PD_n(1+\omega_n(X^{(k)}-1))$; %see \cite{GlassermanLi2005}, original conditional PDs correspond to $\†au_y=0.$
    \item Calculate the portfolio loss $L^{(k)}=\sum_{n=1}^N a_n \LGD_n^{(k)} Y_n^{(k)}$;
%    \item Set $P^{(k)} = \bar{r}_{\tau_y,t}(L^{(k)};X^{(k)})\cdot L^{(k)}$;
\end{enumerate}
}
\begin{enumerate}
  \setcounter{enumi}{6}
\item Sort the values $(L^{(k)})_{k=1,\dots,K}$ in increasing order and obtain $(L^{(k_j)})_{j=1,\dots,K}$ ;
\item Determine $i^* = 	\inf \left\{i \in \N~|~ \sum_{j=1}^i \bar{r}_{\tau_y,t}(L^{(k_j)};X^{(k_j)})\cdot L^{(k_j)} \geq q\cdot K \right\};$
\end{enumerate}
\Output{Return estimator $L^{(k_{i^*})}$ of $\alpha_q(L)$;}%following the approach outlined in \cite{glynn1996importance};} 
\hspace{0.5cm}

\caption{IS algorithm for VaR estimation in the CreditRisk$^+$ model.}\label{algo IS CR+}
\end{footnotesize}}
\end{algorithm}

\begin{algorithm}[h!]

{\begin{footnotesize}
\SetAlgoLined
%\KwData{Data;}
\SetKwInOut{Input}{Input}
\SetKwInOut{Output}{Output}

\Input{Number $K$ of simulations}
\begin{enumerate}
\item Determine $c$ as in \eqref{equ t(c,X) in CM model};
\item Determine $\mu$ as in \eqref{mu_x};
\end{enumerate}
\For{$k=1,\dots, K$}{
\begin{enumerate}
  \setcounter{enumi}{2}
    \item Generate $X^{(k)} \sim\mathcal{N}(\mu,1)$ and determine $t(c,X^{(k)})$ as in equation (\ref{equ t(c,X) in CM model});
     \item Generate $\LGD_n^{(k)}$ according to a Beta-distribution with mean $\operatorname{ELGD}_n$ and variance $\nu\cdot \operatorname{ELGD}_n \cdot (1-\operatorname{ELGD}_n)$ for $n=1,\dots,N$;
    \item Generate a realization of $(Z_1^{(k)},\ldots,Z_N^{(k)})$ with defaults conditionally independent given $X^{(k)}$ and state $s$ probabilities of borrower $n$ given by
    $$
    \frac{\exp(-ta_n \frac{P_{nT}(s)}{P_{n0}})\pi_{ns}(X^{(k)})}{\sum_{s=1}^S e^{-ta_n \frac{P_{nT}(s)}{P_{n0}}}\pi_{ns}(X^{(k)})};
    $$
    \item Calculate the portfolio loss
$$ L^{(k)}= -\sum_{n=1}^N a_n \sum_{s=1}^S \frac{P_{nT}(s)}{P_{n0}} Z_n^{(k)}(s)$$
\end{enumerate}
}
\begin{enumerate}
  \setcounter{enumi}{6}
\item Sort the values $(L^{(k)})_{k=1,\dots,K}$ in increasing order and obtain $(L^{(k_j)})_{j=1,\dots,K}$ ;
\item Determine ${\footnotesize i^* = 	\inf \left\{ i \in \N ~|~ \sum_{j=1}^i \left( M_{L|X^{(k_j)}}(t) \cdot r_{\mu(X^{(k_j)})} \cdot e^{-t L^{(k_j)}} \right) \cdot L^{(k_j)}\geq q\cdot K \right\}}$
    with
    $
    r_{\mu(X^{(k)})}=\exp(-\mu X^{(k)}+\frac{1}{2}\mu^2);
    $

\end{enumerate}
\Output{Return estimator $L^{(k_{i^*})}$ of $\alpha_q(L)$;}
\hspace{0.5cm}
%\Output{Determine full IS estimator as}
%    $$
%    \hat{\theta}_K^{IS}=\frac{1}{K} \sum_{k=1}^K r_{\mu}(X^{(k)}) \hat{\theta}_{n_1}^{IS,1}(X^{(k)})
%    $$

\caption{IS algorithm for ratings-based CreditMetrics model.}\label{algo IS CM}

\end{footnotesize}}
\end{algorithm}

\begin{algorithm}[h]
{\begin{footnotesize}
\SetAlgoLined
%\KwData{Data;}
\SetKwInOut{Input}{Input}
\SetKwInOut{Output}{Output}

\Input{Distributions to sample $\PD, ELGD, a$ and $\omega$; Distribution to sample the number of obligors $N_{\operatorname{Obligors}}$;  Maximal number of obligors $\overline{N}$;  Number of iterations $N_{\operatorname{Iter}} \in \N$; Quantile level $q\in [0,1]$;  Hyperparameters of the NN;}
Initialize the neural network $GA^{\rm NN,act} \in  \mathfrak{N}^{\varphi}_{4N+1,1}$;\\
\For{$\operatorname{iteration} =1,\dots, N_{\operatorname{Iter}}  $}{
Sample the number of obligors $N_{\operatorname{Obligors}}\in \N$;\\
Sample $a = (a_{j})_{ j=1,\dots,N_{\operatorname{Obligors}}}$, $\operatorname{ELGD} = (\operatorname{ELGD}_{j})_{j=1,\dots,N_{\operatorname{Obligors}}}$, $\operatorname{PD} = (\operatorname{PD}_{j})_{j=1,\dots,N_{\operatorname{Obligors}}}$;, $\omega = (\omega_{j})_{j=1,\dots,N_{\operatorname{Obligors}}}$;\\
Compute the quantile loss $Q^L$ according to Algorithm~\ref{algo IS CR+};
 \\
 Compute the conditional loss $C^L = \sum_{n=1}^N a_n \ELGD_n \pi_n(\alpha_q(X))$; \\
Compute the granularity adjustment $GA = Q^L - C^L \in \R$; \\
Compute $GA^{\rm 1st,act}$ as in \eqref{equ GA 1st order actuarial};\\
Minimize 
\[
\bigg(GA^{\rm NN,act} \left(\widetilde{a}, \widetilde{\operatorname{PD}}, \widetilde{\operatorname{ELGD}} ,\widetilde{\omega},GA^{\rm 1st,act} \right) - GA\bigg)^2
\]
with the backpropagation algorithm, where 
$\widetilde{a} = \left(\left( a_{j} \right)_{j=1,\dots,N_{\operatorname{Obligors}}},0, \cdots,0 \right) \in \R^{\overline{N}}$, $\widetilde{\operatorname{PD}} = \left(\left({\operatorname{PD}}_{j}\right)_{j=1,\dots,N_{\operatorname{Obligors}}}, 0, \cdots,0 \right) \in \R^{\overline{N}}$, $ \widetilde{\operatorname{ELGD}} =\left(\left({\operatorname{ELGD}}_{j}\right)_{j=1,\dots,N_{\operatorname{Obligors}}}, 0, \cdots,0 \right) \in \R^{\overline{N}}$, $ \widetilde{\omega} =\left(\left(\omega_{j}\right)_{j=1,\dots,N_{\operatorname{Obligors}}}, 0, \cdots,0 \right) \in \R^{\overline{N}}$ ;

}
\Output{Neural network $GA^{\rm NN,act}$;}
 \caption{Training of a NN for the GA in the actuarial approach}\label{algo_actuarial}
 \end{footnotesize}}
\end{algorithm}

\begin{algorithm}[h!]
{\begin{footnotesize}
\SetAlgoLined
%\KwData{Data;}
\SetKwInOut{Input}{Input}
\SetKwInOut{Output}{Output}

\Input{Distributions to sample $\ELGD, a,\rho,c, \tau, g $; Distribution to sample the number of obligors $N_{\operatorname{Obligors}}$;  Transition Matrix $(p_{s,s,})_{s,s' =0,\dots,S}$; Maximal number of obligors $\overline{N}$; Number of iterations $N_{\operatorname{Iter}} \in \N$;  Quantile level $q\in [0,1]$; }
Initialize the neural network $GA^{\rm NN, MtM} \in \mathfrak{N}^{\varphi}_{6N+1,1}$;\\
\For{$\operatorname{iteration} =1,\dots, N_{\operatorname{Iter}}  $}{
Sample number of obligors $N_{\operatorname{Obligors}} \in \N$;\\ Sample $a = (a_{j})_{j=1,\dots,N_{\operatorname{Obligors}}}$, $\operatorname{ELGD} = (\operatorname{ELGD}_{j})_{j=1,\dots,N_{\operatorname{Obligors}}}$, $\rho = (\rho_{j})_{j=1,\dots,N_{\operatorname{Obligors}}}$, $\operatorname{c} = (\operatorname{c}_{j})_{j=1,\dots,N_{\operatorname{Obligors}}}$,  $(g_{j})_{j=1,\dots,N_{\operatorname{Obligors}}}$, and $\tau = (\operatorname{\tau}_{j})_{ j=1,\dots,N_{\operatorname{Obligors}}}$;\\
Compute the quantile loss $Q^L$ with Algorithm~\ref{algo IS CM}; \\
Compute the conditional loss $C^L = -\sum_{n=1}^N a_n \cdot \mu_n(\alpha_{1-q}(X))$;\\
Compute the granularity adjustment $GA =e^{-rT} \left( Q^L - C^L \right)\in \R$; \\
Compute $GA^{\rm 1st, MtM}$ as in \eqref{eq_GA_formula_1st_order_mtm};\\
Minimize 
\[
\bigg(GA^{\rm NN, MtM}\left(\widetilde{a}, \widetilde{\operatorname{ELGD}} ,\widetilde{\rho} ,\widetilde{\operatorname{c}},\widetilde{g},\widetilde{\tau},GA^{\rm 1st, MtM} \right) - GA\bigg)^2
\]
with the backpropagation algorithm, where 
$\widetilde{a} = \left( \left(a_{j}\right)_{j=1,\dots,N_{\operatorname{Obligors}}},0, \cdots,0 \right) \in \R^{\overline{N}}$,
$ \widetilde{\operatorname{ELGD}} =\left(\left({\operatorname{ELGD}}_{j}\right)_{j=1,\dots,N_{\operatorname{Obligors}}}, 0, \cdots,0 \right) \in \R^{\overline{N}}$, $ \widetilde{\rho} =\left(\left(\rho_{j}\right)_{j=1,\dots,N_{\operatorname{Obligors}}}, 0, \cdots,0 \right) \in \R^{\overline{N}}$ ,
$\widetilde{\operatorname{c}} = \left(\left({\operatorname{c}}_{j}\right)_{j=1,\dots,N_{\operatorname{Obligors}}}, 0, \cdots,0 \right) \in \R^{\overline{N}}$,
$\widetilde{\operatorname{g}} = \left(\left({\operatorname{g}}_{j}\right)_{j=1,\dots,N_{\operatorname{Obligors}}}, 0, \cdots,0 \right) \in \R^{\overline{N}}$, 
$\widetilde{\tau} = \left(\left({\tau}_{j}\right)_{j=1,\dots,N_{\operatorname{Obligors}}}, 0, \cdots,0 \right) \in \R^{\overline{N}}$; 

}
\Output{Neural network $GA^{\rm NN, MtM}$;}
 \caption{Training of a NN for the GA in the MtM approach}\label{algo_MtM}
    
\end{footnotesize}}
\end{algorithm}

\end{document}